\documentclass[11pt]{article}
\usepackage{fullpage}
\usepackage{subfig}
\usepackage[usenames,dvipsnames]{xcolor}
\usepackage[colorlinks,citecolor=blue,linkcolor=BrickRed]{hyperref}
	\usepackage{latexsym,graphicx,epsfig,color}
\usepackage{amsfonts,amssymb,amsmath,amsthm,amstext}
\usepackage{libertine}
\usepackage[libertine]{newtxmath}
\usepackage{enumitem}
\setitemize{itemsep=2pt,topsep=0pt,parsep=0pt}
\usepackage{url,setspace}
\usepackage{multirow}
\usepackage{rotating}
\usepackage{makeidx}
\usepackage{tikz}
\usepackage{accents}
\usepackage{xspace}
\usepackage{algorithm,algpseudocode}
\usepackage{bm}
\usepackage{nicefrac}
\usepackage{changepage} 	%\begin{adjustwidth}{0.5cm}{} \end{adjustwidth}
\usepackage{thmtools,thm-restate}
\usepackage[compact]{titlesec}


\newcommand{\IGNORE}[1]{}
\allowdisplaybreaks

\usetikzlibrary{decorations.markings}
\usetikzlibrary{arrows}
\tikzstyle{block}=[draw opacity=0.7,line width=1.4cm]
\tikzstyle{graphnode}=[circle, draw, fill=black!20, inner sep=0pt, minimum width=6pt]
\tikzstyle{point}=[circle, draw, fill=black!30, inner sep=0pt, minimum width=1pt]
\tikzstyle{input}=[rectangle, draw, fill=black!75,inner sep=3pt, inner ysep=3pt, minimum width=4pt]
\tikzstyle{unmatched}=[graphnode,fill=black!0]
\tikzstyle{shaded}=[graphnode,fill=black!20]
\tikzstyle{matched}=[graphnode,fill=black!100]  	
\tikzstyle{matching} = [ultra thick]
\tikzset{
    >=stealth',
    pil/.style={
           ->,
           thick,
           shorten <=2pt,
           shorten >=2pt,}
}
\tikzset{->-/.style={decoration={
  markings,
  mark=at position .5 with {\arrow{>}}},postaction={decorate}}}
\makeindex

\newtheorem{theorem}{Theorem}[section]
\newtheorem{claim}[theorem]{Claim}

\newtheorem{lemma}[theorem]{Lemma}
\newtheorem{corollary}[theorem]{Corollary}

\newtheorem{fact}[theorem]{Fact}
\theoremstyle{definition}

\newtheorem{definition}[theorem]{Definition}

\usepackage{thmtools,thm-restate} % Used to restate lemmas without new introducing new numbers. 

\def\bz {{z}}

\def\eps {\varepsilon}

\newcommand{\poly}{\operatorname{poly}}
\newcommand{\opt}{\textit{opt\/}}

\usepackage[colorinlistoftodos,textsize=tiny,textwidth=2cm,color=red!25!white,obeyFinal]{todonotes}

\title{Non-clairvoyant Precedence Constrained Scheduling\footnote{This research was supported in part by NSF awards CCF-1536002, CCF-1540541, and CCF-1617790, and the Indo-US Joint Center for Algorithms Under Uncertainty. Sahil Singla was supported in part by the Schmidt Foundation.}}
\date{ \today}
	\author{Naveen Garg\thanks{
	(naveen@cse.iitd.ac.in)
        Computer Science and Engineering Department,
        Indian Institute of Technology, Delhi.
	}
	\and Anupam Gupta\thanks{
        (anupamg@cmu.edu)
        Computer Science Department,
        Carnegie Mellon University.
        }
    \and Amit Kumar\thanks{
        (amitk@cse.iitd.ac.in)
        Computer Science and Engineering Department,
        Indian Institute of Technology, Delhi.
        }
	\and Sahil Singla\thanks{
        (singla@cs.princeton.edu)
        Department of Computer Science at
        Princeton University and
        School of Mathematics at Institute for Advanced Study, Princeton.   
             }
}

\newcommand{\gin}{\gamma^{\text{in}}}
\newcommand{\chain}{{\small\mathsf{chain}}}
\newcommand{\gout}{\gamma^{\text{out}}}

\newcommand{\A}{{\cal A}}
\newcommand{\B}{{\cal B}}
\newcommand{\tz}{{\tilde z}}
\newcommand{\calS}{\mathcal{S}}

\newcommand{\bL}{{\bar L}}
\renewcommand{\bz}{{\bar z}}
\newcommand{\bR}{{\bar R}}
\newcommand{\btR}{\bR^t}
\newcommand{\btz}{\bz^t}
\newcommand{\btL}{\bL^t}
\newcommand{\thetat}{\theta^t}
\newcommand{\etat}{\eta^t}
\newcommand{\nut}{\nu^t}
\newcommand{\Jact}{J^{\text{act}}}
\newcommand{\ind}[1]{{\mathbf{I}[#1]}}
\newcommand{\wh}{{\widehat w}}
\newcommand{\sse}{\subseteq}
\newcommand{\ts}{\textstyle}
\newcommand{\R}{\mathbb{R}}

\newcommand{\redd}[1]{{\color{red}#1}}

\begin{document}

\maketitle
\thispagestyle{empty}

\begin{abstract}
  We consider the online problem of scheduling jobs on identical
  machines, where jobs have precedence constraints. We are interested in
  the demanding setting where the jobs sizes are not known up-front, but
  are revealed only upon completion (the {\em non-clairvoyant} setting).
  Such precedence-constrained scheduling problems routinely arise in map-reduce
  and large-scale optimization. In this paper, we make progress on this problem.  For the objective of total
  weighted completion time, we give a constant-competitive
  algorithm. And for total weighted flow-time, we give an
  $O(1/\eps^2)$-competitive algorithm under $(1+\eps)$-speed
  augmentation and a natural ``no-surprises'' assumption on release
  dates of jobs (which we show is necessary in this context).

  Our algorithm proceeds by assigning {\em virtual rates} to all the
  waiting jobs, including the ones which are dependent on other
  uncompleted jobs, and then use
  these virtual rates to decide on the actual rates of minimal jobs
  (i.e., jobs which do not have dependencies and hence are eligible to
  run). Interestingly, the virtual rates are obtained by allocating time
  in a fair manner, using a Eisenberg-Gale-type convex program (which we
  can also solve optimally using a primal-dual scheme). The optimality condition of this convex
  program allows us to show dual-fitting proofs more easily, without
  having to guess and hand-craft the duals.  We feel that this idea of
  using fair virtual rates should have broader applicability in scheduling
  problems.
\end{abstract}

\newpage
\setcounter{page}{1}

%\setlength{\abovedisplayskip}{2pt}
%\setlength{\belowdisplayskip}{2pt}

%\begin{abstract}{
%}\end{abstract}

% !TeX root = main.tex 
% !TEX root = main.tex

\newcommand{\wjcjfactor}{10}

\section{Introduction}
\label{sec:introduction}

We consider the problem of online scheduling of jobs under precedence constraints. We seek to minimize the average weighted flow time %(or response time)
of the jobs on multiple parallel machines, in the online 
\emph{non-clairvoyant} setting. Formally, there are $m$
identical machines, each capable of one unit of processing per
unit of time. A set of $[n]$ jobs arrive online. Each job has a
processing requirement $p_j$ and a weight $w_j$, and is released at some
time $r_j$. If the job finishes at time $C_j$, its
flow or response time is defined to be $C_j - r_j$. The goal is to give
a preemptive schedule that minimizes the total (or, equivalently, the
average) weighted flow-time $\sum_{j \in [n]} w_j\cdot (C_j - r_j)$. 
The main constraints of
our model are the following: (i)~the scheduling is done \emph{online},
so  the scheduler does not know of the jobs before they are released;
(ii)~the scheduler is \emph{non-clairvoyant}---when a job arrives, the
scheduler knows its weight but \emph{not its processing time $p_j$}. (It
is only when the job finishes its processing that the scheduler
knows the job is done, and hence knows $p_j$.); And (iii)~there are
\emph{precedence constraints} between jobs given by a partial order
$([n],\prec)$: $j \prec j'$ means job $j'$ cannot be started until $j$
is finished. Naturally, the partial order should respect release dates: if
$j \prec j'$ then $r_j \leq r_j'$. (We will require a stronger
assumption for some of our results.)

This model for constrained parallelism is a  natural one, both in
theory and in practice. In theory, this precedence-constrained
(and non-clairvoyant!) scheduling model (with other objective functions) goes back to Graham's
work on list scheduling~\cite{Graham66}.
%  , who gave an online (and !)
% $2$-competitive algorithm for the minimum makespan problem with
% precedences (giving the ``work plus span'' bound for parallel
% programs) 
In practice, most languages and libraries produce parallel code that can
be modeled using precedence DAGs~\cite{RS,ALLM,GrandlKRAK16}. Often these jobs (i.e., units of processing) are distributed among some $m$
%common 
workstations or servers, either in  server farms or on
the cloud, i.e., they use identical parallel machines.

\subsection{Our Results and Techniques}

\medskip\noindent\textbf{Weighted Completion Time.}  We develop our
techniques on the problem of minimizing the \emph{average weighted
  completion time} $\sum_j w_j C_j$. Our
convex-programming approach gives us:

\begin{theorem}
  \label{thm:main1}
  There is a $\wjcjfactor$-competitive deterministic online algorithm
  for minimizing the average weighted completion time on parallel
  machines with both release dates and precedences, in the online
  non-clairvoyant setting.
\end{theorem}

For this result, at each time $t$, the algorithm has to know only the
partial order restricted to $\{ j \in [n] \mid r_j \leq t\}$, i.e., the
jobs released by time $t$. The algorithmic idea is simple in hindsight:
the algorithm looks at the \emph{minimal} unfinished jobs (i.e., they do
not depend on any other unfinished jobs): call them $I_t$. If $J_t$ is
the set of (already released and) unfinished jobs at time $t$, then $I_t
\sse J_t$. To figure out how to divide our processing among the jobs in
$I_t$, we write a convex program that fairly divides the time among all
jobs in the larger set $J_t$, such that (a)~these jobs can ``donate''
their allocated time to some preceding jobs in $I_t$, and that (b)~the
jobs in $I_t$ do not get more than $1$ unit of processing per time-step.

For this fair allocation, we maximize the (weighted) Nash Welfare
$\sum_{j \in J_t} w_j \log R_j$, where $R_j$ is the \emph{virtual rate}
of processing given to job $j \in J_t$, regardless of whether it can
currently be run (i.e., is in $I_t$). This tries to fairly distribute
the virtual rates among the jobs~\cite{Nash}, and can be solved using an
Eisenberg-Gale-type convex program. (We can solve this convex program in
our setting using a simple primal-dual algorithm,
see~\S\ref{sec:convsolv}.)  The proof of Theorem~\ref{thm:main1} is via
writing a linear-programming relaxation for the weighted completion time
problem, and fitting a dual to it. Conveniently, the dual variables for
the completion time LP naturally fall out of the dual (KKT) multipliers
for the convex program!

\medskip\noindent\textbf{Weighted Flow Time.}  We then turn to the
\emph{weighted flow-time minimization} problem. We first observe that
the problem has no competitive algorithm if there are jobs $j$ that 
depend on jobs released before $r_j$. Indeed, if OPT
ever has an empty queue while the algorithm is processing jobs, the
adversary could give a stream of tiny new jobs, and we would be sunk.
%(We formalize this in \S\ref{sec:lower-bounds}.)
Hence we make an additional \emph{no-surprises} assumption about our
instance: when a job $j$ is released, all the jobs having a precedence
relationship to $j$ are also released at the same time. In other words,
the partial order is a collection of disjoint connected DAGs, where all
jobs in each connected component have the same release date. 
A special case of
 this model has been
 studied in \cite{RS,ALLM} where each DAG is viewed as a ``hyper-job''
 and there are no precedence constraints between different hyper-jobs.
% This model
% has been used, e.g., by~\cite{ALLM,GrandlKRAK16}. \alert{Check!}
In this model, we show:
% =======
% released later, which depend on jobs released in the past. Indeed, if
% there is ever a time where OPT has an empty queue, and the algorithm is
% still processing jobs, the adversary could give a stream of new jobs,
% and we would be sunk.  Hence we make an additional assumption about our
% instance, which we call the \emph{no-surprises} assumption: when a job
% $j$ is released, all the jobs having a precedence relationship to $j$
% are also released at the same time. In other words, the partial order is
% a collection of disjoint connected DAGs, where all jobs in each
% connected component have the same release date. A special case of
% this model has been
% studied in \cite{RS,ALLM} where each DAG is viewed as a ``hyper-job''
% and there are no precedence constraints between different hyper-jobs.
% See further discussion in \S\ref{sec:related-work}.
% In this model, we give
% the following theorem.
% \agnote{Name? I am worried about calling it the DAG model,
%   because that's what ALLM call it, and it will get confused.}

\begin{theorem}
  \label{thm:flow1}
  There is an $O(1/\eps^2)$-competitive deterministic non-clairvoyant
  online algorithm for the problem of minimizing the average weighted
  flow time on parallel machines  with release dates and precedences,
  under the no-surprises and  $(1+\eps)$-speedup assumptions.
\end{theorem}

% (Moreover, in \S\ref{sec:lower-bounds} we show a strong lower bound
% without the no-surprises assumption.)  
Interestingly, the algorithm for weighted flow-time is almost the same
as for % Theorem~\ref{thm:main1} for
weighted completion time. In fact, \emph{exactly} the same algorithm
works for both the completion time and flow time cases, if we allow a
speedup of $(2+\eps)$ for the latter. To get the $(1+\eps)$-speedup
algorithm, we give preference to the recently-arrived jobs, since they
have a smaller current time-in-system and each unit of waiting
proportionally hurts them more. This is along the lines of strategies
like LAPS and WLAPS~\cite{EdmondsP}.

%\medskip
%\noindent\textbf
\subsection{The Intuition} 

Consider the case of unit
weight jobs on a single machine. Without precedence constraints, the
round-robin algorithm, which runs all jobs at the same rate, is
$O(1)$-competitive for the flow-time objective with a $2$-speed
augmentation. Now consider precedences, and let the partial order be a
collection of disjoint chains: only the first remaining job from each
chain can be run at each time.
%We call the jobs that do not depend on any
%other existing jobs the \emph{independent} jobs. 
We generalize round-robin to this setting by running all minimal jobs
simultaneously, but at rates proportional to length of the corresponding
chains. We can show this algorithm is also $O(1)$-competitive with a
$2$-speed augmentation. While this is easy for chains and trees, let us
now consider the case when the partial order is the union of general
DAGs, where each DAG may have several minimal jobs. Even though the sum
of the rates over all the minimal jobs in any particular DAG should be
proportional to the number of jobs in this DAG, running all minimal jobs
at equal rates does not work. (Indeed, if many jobs depend on one of
these minimal jobs, and many fewer depend on the other minimal jobs in
this DAG, we want to prioritize the former.) 

Instead, we use a convex program to find rates. Our approach assigns a
``virtual rate'' $R_j$ to each job in the DAG (regardless of whether it
is minimal or not). This virtual rate allows us to ensure that even
though this job may not run, it can help some minimal jobs to run at
higher rates. This is done by an assignment problem where these virtual
rates get translated into actual rates for the minimal jobs. The virtual
rates are then calculated using Nash fairness, which gives us max-min
properties that are crucial for our analysis.

\medskip
\noindent\emph{Analysis Challenges:} In
typical applications of the dual-fitting technique, the dual variables
for each job encode the {\em increase in total flow-time} caused by
arrival of this job. Using this notion turns out to create problems.
Indeed, consider a minimal job of low weight which is running at a high
rate (because a large number of jobs depend on it). The increase in
overall flow-time because of its arrival is very large. However the dual
LP constraints require these dual variables to be bounded by the weights
of their jobs, which now becomes difficult to ensure. To avoid this, we
define the dual variables directly in terms of the virtual rates of the
jobs, given by the convex program.

Having multiple machines instead of a single machine creates new
problems. The \emph{actual rates} assigned to any minimal job cannot
exceed $1$, and hence we have to throttle certain actual rates. Again
the versatility of the convex program helps us, since we can add this as
a constraint. Arguing about the optimal solution to such a convex
program requires dealing with the suitable KKT conditions, from which we
can infer many useful properties. We also show
in \S\ref{sec:convsolv} that the optimal solution corresponds to a
natural ``water-filling'' based algorithm.

Finally, we obtain matching results for the case of $(1+\eps)$-speed
augmentation. Im et al.~\cite{ImKM18} gave a general-purpose technique
to translate a round-robin based algorithm to a LAPS-like algorithm.
% This allows us to convert the $2$-speed augmentation assumption to
% $(1+\eps)$-speed augmentation. 
In our setting, it turns out that the
LAPS-like policy needs to be run on the virtual rates of jobs. Analyzing
this algorithm does not follow in a black-box manner (as prescribed
by~\cite{ImKM18}), and we need to adapt our
dual-fitting analysis suitably.
% has been
% well studied both in the online and schedul

% We are given $n$ jobs $j_1, \ldots, j_n$. Each job $j$ has a processing
% size $p_j$, which is not known to the algorithm -- this is the
% non-clairvoyant model, where the algorithm learns about the size of a
% job only upon its completion. Moreover, we are given a DAG $G$ on the jobs which denotes the precedence relation between them.
% There are $m$ parallel machines.

% In the first problem, the jobs arrive in an on-line manner. Let $r_j$ denote the release date of job $j$. Wlog we assume that all jobs which precede $j$ have release dates at most $r_j$. When the job arrives, it reveals all its predecessors in the DAG $G$. We assume that each job $j$ also has a weight $w_j$.
% Our objective is to process the jobs in order to  minimize the total weighted completion time of jobs.

% In the second problem, DAGs arrive in an on-line manner. All jobs in a DAG arrive at the same time. As before each job $j$ has a weight $w_j$, and a  processing size $p_j$; but the algorithm is non-clairvoyant. The objective is to minimize the total weighted flow-time with 2-speed augmentation. 

\subsection{Related Work and Organization}
\label{sec:related-work}

\noindent\textbf{Completion Time.}~ Minimizing $\sum_j w_jC_j$
on parallel machines with precedence constraints has
$O(1)$-approximations in the \emph{offline} setting: Li~\cite{Li17}
improves on~\cite{Hall,MQS} to give a $3.387+\eps$-approximation. For
\emph{related} machines, the precedence constraints make the problem
much harder: there is a $O(\log m/\log \log
m)$-approximation~\cite{Li17} improving on a prior $O(\log m)$
result~\cite{ChudakS99}, and a hardness of $\omega(1)$ under
certain complexity assumptions~\cite{BNF15}. In the \emph{online}
setting, any offline algorithm for (a dual problem to) $\sum_j w_jC_j$
gives an \emph{clairvoyant} online algorithm, losing $O(1)$
factors~\cite{Hall}. % Hence, one gets a online algorithm for $\sum_j
% w_jC_j$ with precedence constraints.
Two caveats: it is unclear (a) how to make this algorithm
non-clairvoyant, and (b) how to solve the (dual of the) weighted
completion time problem with precedences in poly-time.

\medskip\noindent\textbf{Flow Time without Precedence.}~ To minimize
$\sum_j w_j(C_j-r_j)$, %has been extensively studied in the recent past.
strong lower bounds are known for the competitive ratio of any online
algorithm even on a single machine~\cite{MPT}. Hence we use speed
augmentation~\cite{KP00}. % introduced speed augmentation in
% the setting of non-clairvoyant (unweighted) flow-time on a single
% machine.
For the general setting of non-clairvoyant weighted flow-time on
unrelated machines, Im et al.~\cite{ImKMP14} showed that weighted
round-robin with a suitable migration policy yields a
$(2+\eps)$-competitive algorithm using $(1+\eps)$-speed
augmentation. They gave a general purpose technique, based on the LAPS
scheduling policy, to convert any such round-robin based algorithm to a
$(1+\eps)$-competitive algorithm while losing an extra $1/\eps$ factor
in the competitive ratio. Their analysis also uses a dual-fitting
technique~\cite{AGK12,GuptaKP12}. However, they do not consider
precedence constraints.

\medskip\noindent\textbf{Flow Time with Precedence.}~
%In contrast to the list above, much less is
Much less is known for flow-time problems with precedence constraints.
For the offline setting on identical
machines, \cite{KL18} % Now no poly-time algorithm can achieve an
% approximation factor better than $n^{1/(\log \log n)^c}$ for some
% constant $c$ withotu . 
give $O(1)$-approximations with $O(1)$-speedup, even for general delay
functions. %\agnote{They allow general release dates.}
In the current paper, we achieve a $\poly(1/\eps)$-approximation with
$(1+\eps)$-speedup for flow-time.  Interestingly, \cite{KL18} show that
beating a $n^{1-c}$-approximation for any constant $c \in [0,1)$
requires a speedup of at least the optimal approximation factor of
makespan minimization in the same machine environment. 
% There is no
% contradiction between this and our algorithm, since 
However, this lower bound requires different jobs with a precedence
relationship to have different release dates, which is something our
model disallows. (Appendix~\S\ref{sec:lower-bounds} gives another lower bound
% in a similar vein, 
showing why we disallow such precedences in the online setting.)
% \alert{stopping here. continue
%   around 3:30pm}

In the \emph{online} setting, \cite{RS} introduced the DAG model where
each \emph{job} is a directed acyclic graph (of \emph{tasks}) released
at some time, and a job/DAG completes when all the tasks in it are
finished, and we want to minimize the total \emph{unweighted} flow-time.
They gave a $(2+\eps)$-speed $O(\kappa/\eps)$-competitive algorithm,
where $\kappa$ is the largest antichain within any job/DAG. \cite{ALLM}
show $\poly(1/\eps)$-competitiveness with $(1+\eps)$-speedup, again in
the non-clairvoyant setting. The case where jobs are entire DAGs, and
not individual nodes within DAGs, is captured in our weighted model by
putting zero weights for all original jobs, and adding a unit-weight
zero-sized job for each DAG which now depends on all jobs in the
DAG. Assigning arbitrary weights to individual nodes within DAGs makes
our problem quite non-trivial---we need to take into account the
structure of the DAG to assign rates to jobs. Another model to capture
parallelism and precedences uses \emph{speedup functions}~\cite{EdmondsCBD97,Edmonds99,EdmondsP}: 
relating our model to this setting remains an open question.

Our work is closely related to Im et al.~\cite{ImKM18} who use a Nash
fairness approach for completion-time and flow-time problems with
multiple resources. While our approaches are similar, to the best of our
understanding their approach does not immediately extend to the setting
with precedences. Hence we have to introduce new ideas of using virtual
rates (and being fair with respect to them), and throttling the induced
actual rates at $1$. The analyses of~\cite{ImKM18} and our work are both
based on dual-fitting; however, we need
some new ideas for the setting with precedences.
%%We consider the \emph{weighted} case: this
%%changes the problem considerably, as the following example
%%shows.
%\agnote{Blah}

%\alert{Other papers: Im-Kulkarni-Munagala, .}

\medskip\noindent\textbf{Organization.} The weighted completion time
case is solved in \S\ref{sec:compleTime}. A $(2+\eps)$-speedup result
for weighted flow-time is in \S\ref{sec:flow}; this is improved to a
$(1+\eps)$-speedup in \S\ref{sec:improve-flowtime}. The proof that we
need the ``no-surprises'' assumption on release dates is in
\S\ref{sec:lower-bounds}.  Finally, we show how to solve the convex
program in \S\ref{sec:convsolv}. Some deferred proofs can be found in 
\S\ref{sec:proofs}.

%%% Local Variables:
%%% mode: latex
%%% TeX-master: "main"
%%% End:

% !TeX root = main.tex 
% !TEX root = main.tex

\section{Minimizing Weighted Completion Time}
\label{sec:compleTime}

In this section, we describe and analyze the scheduling algorithm for
the problem of minimizing weighted completion time on parallel
machines. Recall that the precedence constraints are given by a DAG $G$,
and each job $j$ has a release date $r_j$, processing size $p_j$ and
weight $w_j$.
%We first describe the scheduling algorithm.

\subsection{The Scheduling Algorithm}
\label{sec:sched-algo}

We first assume that each of the $m$ machines run at rate $2$ (i.e.,
they can perform 2 units of processing in a unit time). We will show
later how to remove this assumption (at a constant loss of competitive
ratio). We begin with some notation. We say that a job $j$ is {\em
  waiting} at time $t$ (with respect to a schedule) if $r_j \leq t$, but
$j$ has not been processed to completion by time $t$. We use $J_t$ to
denote the set of waiting jobs at time $t$. Note that at time $t$, the
algorithm gets to see the subgraph $G_t$ of $G$ which is induced by the
jobs in $J_t$. We say that a job $j$ is {\em unfinished} at time $t$ if
it is either waiting at time $t$, or its release date is at least $t$
(and hence the algorithm does not even know about this job). Let $U_t$
denote the set of unfinished jobs at time $t$. Clearly,
$J_t \subseteq U_t$. At time $t$, the algorithm can only process those
jobs in $J_t$ which do not have a predecessor in $G_t$ -- denote these
\emph{minimal} jobs by $I_t$: they are {independent} of all other current jobs. For every time $t$, the
scheduling algorithm needs to assign a rate to each job $j \in I_t$. We
now describe how it decides on these rates.

Consider a time $t$. The algorithm considers a bipartite graph
$H_t = (I_t, J_t, E_t)$ with vertex set consisting of the minimal jobs
$I_t$ on left and the waiting jobs $J_t$ on right. Since $I_t \sse J_t$,
a job in $I_t$ appears as a vertex on both sides of this bipartite
graph. When there is no confusion, we slightly overload terminology by
referring to a job as a vertex in $H_t$.  The set of edges $E_t$ are as
follows: let $j_l \in I_t, j_r \in J_t$ be vertices on the left and the
right side respectively. Then $(j_l, j_r)$ is an edge in $E_t$ if and
only if there is a directed path from $j_l$ to $j_r$ in the DAG $G_t$.

The following convex program now computes the rate for each vertex in
$I_t$. It has variables $z^t_e$ for each edge $e \in E_t$. For each job
$j$ on the left side, i.e., for $j \in I_t$, define
$L_j^t := \sum_{e \in \partial j} z^t_e$ as the sum of $z_e$ values of
edges incident to $j$. Similarly, define
$R_j^t := \sum_{e \in \partial j} z^t_e$ for a job $j \in J_t$, i.e., on
the right side. The objective function is the Nash bargaining objective
function on the $R_j^t$ values, which ensures that each waiting job gets
some attention. In \S\ref{sec:convsolv} we give a combinatorial algorithm to efficiently solve 
this convex program.
\begin{alignat}{2}
\label{cp}
  \max & \sum_{j \in J_t} w_j \ln R^t_j \tag{CP} \\
  \label{eq:c1}
  L^t_j & =  \sum_{j' \in J_t: (j,j') \in E_t} z^t_{jj'} & \quad \quad & \forall j \in I_t \\
  \label{eq:c2}
  R^t_j & =  \sum_{j' \in I_t: (j',j) \in E_t} z^t_{j'j} & \quad \quad & \forall j \in J_t \\
  \label{eq:c3}
  L^t_j & \leq  1 & \qquad \qquad &\forall j \in I_t \\
  % R_j & \leq  1 & \qquad \qquad &\forall j \in J_t \\
  \label{eq:c4}
   \sum_{j \in I_t} L_j^t & \leq m \\
  \label{eq:c5}
  z^t_e & \geq 0 & \quad \quad & \forall e \in E_t
\end{alignat}
Let $(\btz, \btL, \btR)$ be an optimal solution to the above convex
program. We define the \emph{rate of a job} $j \in I_t$ as being $\btL_j$.

Although we have defined this as a continuous time process, it is easy
to check that the rates only change if a new job arrives, or if a job
completes processing. Also observe that we have effectively combined the
$m$ machines into one in this convex program. But assuming that all
events happen at integer times, % \agnote{Can we remove this integer times
  % assumption cleanly, not sure that things are rational, so need to be
  % careful?} 
we can translate the rate assignment to an actual schedule
as follows. For a time slot $[t, t+1]$, the total rate is at most $m$
(using~(\ref{eq:c4})), so we create $m$ time slots $[t,t+1]_i$, one for
each machine $i$, and iteratively assign each job $j$ an interval of
length $\btL_j$ within these time slots. It is possible that a job may
get assigned intervals in two different time slots, but the fact that
$\btL_j \leq 1$ means it will not be assigned the same time in two
different time slots. Further, we will never exceed the slots because
of~\eqref{eq:c4}. Thus, we can process these jobs in the $m$ time slots
on the $m$ parallel machines such that each job $j$ gets processed for
$\btL_j$ amount of time and no job is processed concurrently on multiple
machines. This completes the description of the algorithm; in this, we
assume that we run the machines at twice the speed. Call this algorithm
$\A$.

The final algorithm $\B$, which is only allowed to run the machines at
speed $1$, is obtained by running $\A$ in the background, and
setting $\B$ to be a slowed-down version of $\A$. Formally, if $\A$
processes a job $j$ on machine $i$ at time $t \in \R_{\geq 0}$, then
$\B$ processes this at time $2t$.
% to an extend of $2 \cdot s^t_j$ during $[t, t+1]$,
% then $\B$ processes this job to the same extent during $[2t,
% 2t+2]$. 
This completes the description of the algorithm.

\subsection{A Time-Indexed LP formulation}
\label{sec:lp}

We use the dual-fitting approach 
to analyze the above algorithm. We write a time-indexed linear programming
relaxation~(\ref{eq:LP}) for the weighted completion time problem, and
use the solutions to the convex program~(\ref{cp}) to obtain feasible
primal and dual solutions for~(\ref{eq:LP}) which differ by only a constant factor.

We divide time into integral time slots (assuming all quantities are integers). Therefore, the variable $t$ will refer to integer times only. For every job $j$ and time $t$, we have a variable $x_{j,t}$ which denotes the volume of $j$ processed during $[t,t+1]$. Note that this is defined only for $t \geq r_j$.
The LP relaxation is as follows:
 \begin{align}
 \min \quad & \ts \sum_{j,t} w_j \cdot \frac{t \cdot
              x_{j,t}}{p_j} \tag{LP} \label{eq:LP} \\
 \label{eq:1}
 \ts \sum_{t \geq r_j} \frac{x_{j,t}}{p_j} & \geq 1 \quad \forall j\\
 \label{eq:2}
 \ts \sum_j x_{j,t} & \leq m \quad \forall t \\
 \label{eq:3}
 \ts \sum_{s \leq t} \frac{x_{j,s}}{p_j} & \geq \ts \sum_{s \leq t} \frac{x_{j',s}}{p_{j'}} \quad \forall t, j \prec j'
\end{align}

The following claim, whose proof is deferred to the appendix,  shows that it is a valid relaxation.
\begin{claim}
\label{cl:lprelx}
Let $\opt$ denote the weighted completion time of an optimal off-line policy (which knows the processing time of all the jobs). Then the optimal value of the LP relaxation is at most $\opt$.
\end{claim}

% \begin{proof}
%   Consider an optimal schedule $\calS$, and let $x_{j,t}$ be the volume of
%   $j$ processed during $[t,t+1]$. Constraint~(\ref{eq:1}) states that
%   the total amount of processing on $j$ must be at least (in fact, it
%   will be equal to) $p_j$. Constraint~(\ref{eq:2}) requires that the
%   total amount of processing that can happen during a slot $[t,t+1]$ is
%   at most $m$ because each machine can perform 1 unit of processing
%   during this time slot. Constraint~(\ref{eq:3}) can be justified as
%   follows: suppose $j$ precedes $j'$, and consider a time $t$. Then the
%   LHS of this constraint denotes the fraction to which $j$ has been
%   processed till time $t$, and the RHS denotes this quantity for
%   $j'$. In the schedule $\calS$, if the RHS is positive, then it must be
%   the case that $j$ has been completed by time $t$, and so the LHS would
%   be 1.
%   Finally, we consider the objective function. Let $C_j$ be the
%   completion time of $j$. Clearly, $x_{j,t} = 0$ for $t > C_j$, and so,
%   $\sum_t \frac{t \cdot x_{j,t}}{p_j} \leq C_j \cdot \frac{\sum_t
%     x_{j,t}}{p_j} = C_j. $
% \end{proof}

The~(\ref{eq:LP}) has a large integrality gap. Observe that the LP just
imagines the $m$ machines to be a single machine with speed $m$. % Also we
% are assuming that the machines are running at speed 1 (and not at speed
% 2) because we are bounding the value of the off-line optimum. \agnote{We
%   don't need the speedup to show the gap instance, right?} 
Therefore,
(\ref{eq:LP}) has a large integrality gap for two reasons: (i)~a job $j$
can be processed concurrently on multiple machines, and (ii)~suppose we
have a long chain of jobs of equal size in the DAG $G$. Then the LP
allows us to process all these jobs at the same rate in parallel on
multiple machines. We augment the LP lower bound with another quantity
and show that the sum of these two lower bounds suffice.

A {\em chain} $C$ in $G$ is a sequence of jobs $j_1, \ldots, j_k$ such that $j_1 \prec j_2 \prec \ldots \prec j_k$. Define the processing time of $C$, $p(C)$, as the sum of the processing time of jobs in $C$. For a job $j$, define
$\chain_j$ as the maximum over all chains $C$ ending in $j$ of $p(C)$. It is easy to see that $\sum_j w_j \cdot (r_j + \chain_j)$ is a lower bound (up to a factor 2) on the objective of an optimal schedule.

We now write down the dual of the LP relaxation above. We have dual variables $\alpha_j$ for every job $j$, and $\beta_t$ for every time $t$, and $\gamma_{s, j \rightarrow j'}$
\begin{align}
 \max \quad  \sum_{j} \alpha_j - m \sum_t \beta_t \label{eq:dp} \tag{DLP}\\
 \label{eq:dual}
 \alpha_j - w_j \cdot t + \sum_{s \geq t} \Big( \sum_{ j \prec j'} \gamma_{s,j \rightarrow j'} - \sum_{ j' \prec j} \gamma_{s, j' \rightarrow j}  \Big) & \leq p_j \cdot \beta_t   \quad \forall j,t \geq r_j \\
 \notag
 \alpha_j, \beta_t & \geq 0
 \end{align}

 We write the dual constraint~\eqref{eq:dual} in a more readable manner. For a job $j$ and time $s$, let $\gin_{s,j}$ denote $\sum_{j' \prec j} \gamma_{s, j' \rightarrow j}$, and define $\gout_{s,j}$ similarly.
We now write the dual constraint~\eqref{eq:dual} as
\begin{align}
\label{eq:dualnew}
 \alpha_j - w_j \cdot t + \sum_{s \geq t} \Big( \gout_{s,j} - \gin_{s,j} \Big) \leq p_j \cdot \beta_t   \quad \forall j,t \geq r_j
\end{align}

\subsection{Properties of the Convex Program}
\label{sec:prop-CP}

We now prove certain properties of an optimal solution
$(\btz, \btL, \btR)$ to the convex program~\eqref{cp}. The first
property, 
whose proof is deferred to the appendix, 
is easy to see:
\begin{claim}
  \label{cl:full}
  If $\sum_{j \in I_t} \btL_j < m$, then $\btL_j = 1$ for all
  $j \in I_t$.
\end{claim}

% \begin{proof}
%   Suppose $\sum_{j \in I_t} \btL_j < m$, but $\btL_j < 1$ for some
%   $j \in I_t$. Let $e$ be an edge incident with $j$ (since there is a
%   copy of $j$ on the right side of the bipartite graph, we know that $j$
%   has at least one edge incident with it). We can raise the $z_e$ value
%   of this edge while maintaining feasibility. But this will increase the
%   objective value, a contradiction.
% \end{proof}

We now write down the KKT conditions for the convex program. (In fact,
we can use~(\ref{eq:c1}) and~(\ref{eq:c2}) to replace $\bL^t_j$ and
$\bR^t_j$ in the objective and the other constraints.) Then letting
$\thetat_j \geq 0, \etat \geq 0, \nut_e \geq 0$ be the Lagrange
multipliers corresponding to constraints~\eqref{eq:c3},~\eqref{eq:c4}
and~\eqref{eq:c5}, we get
\begin{alignat}{2}
  \label{k1}
  \frac{w_j}{\btR_j} & = \theta^t_{j'} + \eta^t - \nu_e^t &\qquad & \forall e=(j',j), j' \in I_t, j \in J_t \\
% \label{k2}
% a_j^t + \thetat_j + \etat & = 0 \quad \forall j \in I_t \\
% \label{k3}
% a^t_j + b^t_{j'}+ \nut_e &  = 0 \quad \forall  \\
\label{k4}
\thetat_j\;(\btL_j - 1) & = 0 && \forall j \in I_t\\
\label{k5}
\ts \etat \;(\sum_{j \in I_t} \btL_j - m) & = 0 && \\
\label{k6}
\nut_e \cdot \btz_e & = 0 && \forall e \in E_t
\end{alignat}

We derive a few consequences of these conditions, the proofs are deferred to the appendix.
\begin{claim}
  \label{cl:sum0}
  Consider a job $j \in J_t$ on the right side of $H_t$. Then
  $w_j \geq \btR_j \cdot \etat$.
\end{claim}
%\begin{proof}
%  Constraint~\eqref{k1} implies that $\btR_j >
%  0$ and so there is a vertex $j' \in I_t$ such that $e=(j',j) \in
%  H_t$ with $\btz_e > 0$. Now~\eqref{k6} shows that $\nut_e =
%  0$, and so $w_j/\bR^t_j = \etat + \theta^t_{j'} \geq
%  \etat$. Hence the proof.
%\end{proof}

\begin{claim}
  \label{cl:sum}
  Consider a job $j \in J_t$ on the right side of $H_t$. Suppose $j$ has
  a neighbor $j' \in I_t$ such that $\btL_{j'} < 1$ and
  $\btz_{j'j} > 0$.  Then $w_j = \btR_j \cdot \etat$.
\end{claim}
%\begin{proof}
%  Let $e$ denote the edge $(j',j)$. Now~\eqref{k4} and~\eqref{k6} imply
%  that $\nut_e = 0$ and $\thetat_{j'} = 0$. The claim now follows
%  from~\eqref{k1}.
%\end{proof}

A crucial notion is that of an \emph{active} job:

\begin{definition}[Active Jobs]
  A job $j \in J_t$ is {\em active} at time $t$ if it has at least one
  neighbor in $I_t$ (in the graph $H_t$) running at rate strictly less
  than 1.
\end{definition}

Let $\Jact_t$ denote the set of active jobs at time $t$.  We can
strengthen the above claim as follows. 
%The proof is deferred to the appendix. 
\begin{corollary}
  \label{cor:sum}
  Consider an active job $j$ at time $t$.  Then
  $w_j = \btR_j \cdot \etat$.
\end{corollary}
% \begin{proof}
%   By definition there is a neighbor $j' \in I_t$ of $j$ such that
%   $\btL_{j'} < 1$.  Let $e'$ denote the edge $(j',j)$. If
%   $\btz_{e'} > 0$, then we are done by Claim~\ref{cl:sum} above. So assume
%   $\btz_{e'} = 0$.  Since $\btR_j > 0$. There must be an edge
%   $e''=(j'',j)$ incident with $j$ such that $\btz_{e''} > 0$.  Again, if
%   $\btL_{j''} < 1$, we are done by the Claim above. So, assume that
%   $\btL_{j''} = 1$. Now consider reducing $\btz_{e''}$ by a tiny amount
%   and increasing $\btz_{e'}$ by the same amount. This maintains
%   feasibility of all constraints. Since $\btR_j$ remains unchanged, we
%   remain at an optimal solution. Now we can apply Claim~\ref{cl:sum}.
% \end{proof}

\begin{claim}
  \label{cl:upper}
  $w(\Jact_t)/m \leq \etat \leq w(J_t)/m$.
\end{claim}

\subsection{Analysis via Dual Fitting}
\label{sec:dual-completion}

We analyze the algorithm $\A$ first.
We  define feasible dual variables for~(\ref{eq:dp}) such that the value of the dual objective function (along with the $\chain_j$ values that capture the maximum processing time over all chains ending in $j$) forms a lower bound on the weighted completion time of our algorithm. Intuitively, $\alpha_j$ would be the weighted completion time of $j$, and $\beta_t$ would be $1/2m$ times the total weight of unfinished jobs at time $t$. Thus, $\sum_j \alpha_j - m \sum_t \beta_t$ would be at $1/2$ times the total weighted completion time. This idea works as long as all the machines are busy at any point of time, the reason being that the primal LP essentially views the $m$ machines as  a single speed-$m$ machine. Therefore, we can generate enough dual lower bound if the rate of processing in each time slot is $m$. If all machines are not busy, we need to appeal to the lower bound given by the $\chain_j$ values.

We use the notation used in the description of the algorithm. In the graph $H_t$, we had assigned rates $\btL_j$ to all the nodes $j$ in $I_t$. Recall that a vertex $j \in J_t$ on the right side of $H_t$ is said to be {\em active} at time $t$ if it has a neighbor $j' \in I_t$ for  which $\btL_{j'} < 1$. Otherwise, we say that
$j$ is inactive at time $t$.
 We say that an edge $e=(j_l, j_r) \in E_t$, where $j_l \in I_t, j_r \in J_t$ is {\em active} at time $t$ if the vertex $j_r$ is active. Let $A_t$ denote the set of active edges in $E_t$.  Let $e = (j_l, j_r)$ be an edge in $E_t$. By definition, there is a path from $j_l$ to $j_r$ in $G_t$ -- we fix such a path $P_e$.
As before, let $C_j$ denote the completion time of job $j$.
The dual variables are defined as follows:
\begin{itemize}
\item For each job $j$ and time $t$, we define quantities $\alpha_{j,t}$. The dual variable $\alpha_j$ would be equal to $\sum_{t \geq 0} \alpha_{j,t}$. Fix a job $j$. If $t \notin [r_j, C_j]$ we set $\alpha_{j,t}$ to 0. Now, suppose $j \in J_t$. Consider the job $j$ as a vertex in $J_t$ (i.e., right side)  in the bipartite graph $H_t$. We set $\alpha_{j,t}$ to $w_j$ if $j$ is active at time $t$,
    otherwise it is inactive.
\item For each time $t$, we set $\beta$ to $1/2m \cdot w(U_t)$ (Recall that $U_t$ is the set of unfinished jobs at time $t$).
\item We now need to define $\gamma_{t, j' \rightarrow j}$, where $j' \prec j$. If $j$ or $j'$ does not belong to $J_t$, we set this variable to 0. So assume that $j, j' \in J_t$ (and so the edge $(j', j)$ lies in $G_t$). We define
$$\gamma_{t, j' \rightarrow j} := \etat \cdot \sum_{e: e \in A_t, (j' \rightarrow j) \in P_e} \btz_e.$$
In other words, we consider all the active edges $e$ in the graph $H_t$ for which the corresponding path $P_e$ contains $(j', j)$. We add up the fractional assignment $\btz_e$ for all such edges.
\end{itemize}

This completes the description of the dual variables.

We first show that the objective function for~(\ref{eq:dp}) is close to the weighted completion time incurred by the algorithm. 
The proof is deferred to the appendix. 
\begin{claim}
\label{cl:obj}
The total weighted completion time of the jobs in $\A$ is at most $2(\sum_j \alpha_j  - m \cdot \sum_t \beta_t) + \sum_j w_j \cdot (\chain_j + 2r_j)$.
\end{claim}
% \begin{proof}
% Fix a job $j$. Let $C$ be the chain in $G$ which ends with $j$ and satisfies $p(C) = \chain_j$. Consider a time $t \leq C_j$, the completion time of $j$. Suppose $\alpha_{j,t} = 0$. Considering $j$ as a vertex in $J_t$ (i.e., right side) in the bipartite graph $H_t$, it must be the case that all its neighbors get rate 1. Exactly one job in the chain $C$, say $j'$, belongs to the set $I_t$. Since $(j',j)$ is an edge in $H_t$, it must be the case that $j'$ gets rate 1. Thus, we conclude that whenever $\alpha_{j,t} = 0$, there is a job in $C$ which is processed for 2 units during $[t,t+1]$ (recall that the machines in $\A$ run at speed 2). Therefore, $w_j(C_j - r_j) \leq \alpha_j + w_j \cdot \chain_j/2$. Summing over all jobs, we get
% $$ \sum_j w_j C_j \leq \sum_j w_j (r_j + \chain_j/2) + \sum_j \alpha_j. $$
% Now observe that for any time $t$, $m \beta_t$ is equal to $w(U_t)/2$, and so, $m \cdot \sum_t \beta_t = \sum_j w_j C_j/2$. Subtracting this from the inequality above yields the desired result.
% \end{proof}

We now argue about feasibility of the dual constraint~\eqref{eq:dual}. Consider a job $j$ and time $t \geq r_j$. Since $\alpha_{j,s} \leq w_j$ for all time $s$, $\sum_{s \leq t} \alpha_{j,s} \leq w_j \cdot t$. Therefore, it suffices to show:
\begin{align}
\label{eq:dualsuf}
\sum_{s \geq t} \alpha_{j,s} + \sum_{s \geq t} \Big(   \gout_{s,j} - \gin_{s,j} \Big)  \leq p_j \cdot \beta_t
\end{align}

Let $t_j^\star$ be the first time $t$ when the job $j$ appears in the set $I_t$. This would also be the first time when the algorithm starts processing $j$ because a job that enters $I_t$ does not leave $I_t$ before completion.
\begin{claim}
\label{cl:tj}
For any time $s$ lying in the range $[r_j, t_j^\star)$, $\alpha_{j,s} + \gout_{s,j} - \gin_{s,j} = 0$.
\end{claim}
\begin{proof}
Fix such a time $s$. Note that $j \notin I_s$. Thus $j$ appears as a vertex on the right side in the bipartite graph $H_s$, but does not appear on the left side.
Let $e$ be in active edge in $H_s$ such that the corresponding path $P_e$ contains $j$ as an internal vertex. Then $\bz_e^s$ gets counted in both $\gout_{s,j}$ and $\gin_{s,j}$. There cannot be such a path $P_e$ which starts with $j$, because then $j$ will need to be on the left side of the bipartite graph. There could be paths $P_e$ which end with $j$ -- these will correspond to active edges $e$ incident with $j$ in the graph $H_t$ (this happens only if $j$ itself is active). Let $\Gamma(j)$ denote the edges incident with $j$. We have shown that
\begin{align}
\label{eq:gin}
\gout_{s,j} - \gin_{s,j} = - \etat \cdot \sum_{e \in \Gamma(j) \cap A_s} \tz^s_e.
\end{align}
If $j$ is not active, the RHS is 0, and so is $\alpha_{j,s}$. So we are done. Therefore, assume that
$j$ is active. Now, $A(s)$ contains all the edges incident with $j$, and so, the RHS is same as
$-\etat \cdot \btR_j$. But then, Corollary~\ref{cor:sum} implies that $- \etat \cdot \btR_j = -w_j$. Since
$\alpha_{j,s} = w_j$, we are done again.
\end{proof}

Coming back to inequality~\eqref{eq:dualsuf}, we can assume that $t \geq t_j^\star$. To see this, suppose $t < t_j^\star$. Then by Claim~\ref{cl:tj} the LHS of this constraint is same as
$$\sum_{s \geq t^\star_j} \alpha_{j,s} + \sum_{s \geq t^\star_j} \Big(   \gout_{s,j} - \gin_{s,j} \Big). $$
Since $\beta_t \geq \beta_{t^\star_j}$ (the set of unfinished jobs can only diminish as time goes on),~\eqref{eq:dualsuf} for time $t$ follows from the corresponding statement for time $t^\star_j$. Therefore, we assume that $t \geq t^\star_j$. We can also assume that $t \leq C_j$, otherwise the LHS of this constraint is  0.

\begin{claim}
\label{cl:aftertj0}
Let $s \in [t_j^\star, C_j]$ be such that $j$ is inactive at time $s$.
Then $\alpha_{j,s} + \gout_{s,j} - \gin_{s,j} \leq \eta^s \cdot \bL^s_j.$
\end{claim}
\begin{proof}
We know  that $\alpha_{j,s} = 0$.
As in the proof of Claim~\ref{cl:tj}, we only need to worry about those active  edges $e$ in $H_s$ for which $P_e$ either ends at $j$  or begins with $j$. Since any edge incident with $j$ as a vertex on the right side is inactive, we get (let $\Gamma(j)$ denote the edges incident with $j$, where we consider $j$ on the left side)
$$\alpha_{j,s} + \gout_{s,j} - \gin_{s,j} = \eta^s \cdot \sum_{e \in \Gamma(j) \cap A(s)} \bz^s_e \leq \eta^s \cdot \bL^s_j, $$
because $\eta^s \geq 0$ and $\bL^s_j = \sum_{e \in \Gamma(j)} \bz^s_e$.
\end{proof}

\begin{claim}
\label{cl:aftertj1}
Let $s \in [t_j^\star, C_j]$ be such that $j$ is active at time $s$.
Then $\alpha_{j,s} + \gout_{s,j} - \gin_{s,j} \leq \eta^s \cdot \bL^s_j. $
\end{claim}
\begin{proof}
The argument is very similar to the one in the previous claim. Since $j$ is active,  $\alpha_{j,s} = w_j$. As before we only need to worry about the active edges $e$ for which $P_e$ either ends or begins with $j$.
Any edge which is incident with $j$ on the right side (note that there will only one such edge -- one the one joining $j$ to its copy on the left side of $H_t$) is active.
The following inequality now follows as in the proof of Claim~\ref{cl:aftertj0}:
$$ \alpha_{j,s} + \gout_{s,j} - \gin_{s,j} \leq w_j + \eta^s \cdot \bL^s_j - \eta^s \cdot \bR^s_j. $$
The result now follows from Corollary~\ref{cor:sum}.
\end{proof}

We are now ready to show that~\eqref{eq:dualsuf} holds.  The above two claims show that the LHS of~\eqref{eq:dualsuf} is at most
$\sum_{s = t}^{C_j} \eta^s \cdot \bL^s_j . $ Note that for any such time $s$, the rate assigned to $j$ is $\bL^j_s$, and so, we perform $2 \cdot \bL^j_s$ amount of processing on $j$ during this time slot. It follows that
$ \sum_{s=t}^{C_j}  \bL^s_j \leq p_j/2$.
Now Claim~\ref{cl:upper} shows that $\eta^s \leq w(U_s)/m \leq
w(U_t)/m$, and so we get
$$\sum_{s=t}^{C_j} \eta^s \cdot \bL^s_j \leq \frac{p_j \cdot w(U_t)}{2m} = p_j \cdot \beta_t. $$
This shows that~\eqref{eq:dualsuf} is satisfied. We can now prove our algorithm is constant competitive.

\begin{theorem}
The algorithm $\B$ is 10-competitive.
\end{theorem}
\begin{proof}
We first argue about $\A$. We have shown that the dual variables are feasible to~(\ref{eq:dp}), and so, Claim~\ref{cl:obj} shows that the total completion time of $\A$ is at most $2 \opt + \sum_j w_j (\chain_j  + 2r_j)$, where $\opt$ denotes the optimal off-line objective value. Clearly, $\opt \geq \sum_j w_j \cdot r_j$ and $\opt \geq \sum_j w_j \cdot \chain_j$. This implies that $\A$ is 5-competitive. While going from $\A$ to $\B$ the completion time of each job doubles.
\end{proof}

%%% Local Variables:
%%% mode: latex
%%% TeX-master: "main"
%%% End:

% !TeX root = main.tex 
% !TEX root = main.tex

\newcommand{\Tnc}{T^{\text{nice}}}
\section{Minimizing Weighted Flow Time}
\label{sec:flow}

We now consider the setting of minimizing the total weighted flow time,
again in the non-clairvoyant setting. The setting is almost the same as
in the completion-time case: the major change is that all jobs which
depend on each other (i.e., belong to the same DAG in the ``collection
of DAGs view'' have the same release date). In \S\ref{sec:lower-bounds}
we show that if related jobs can be released over time then no
competitive online algorithms are possible.

As before, let $J_t$ denote the jobs which are \emph{waiting} at time
$t$, i.e., which have been released but not yet finished, and let $G_t$
be the union of all the DAGs induced by the jobs in $J_t$. Again, let
$I_t$ denote the \emph{minimal} set of jobs in $J_t$, i.e., which do not
have a predecessor in $G_t$ and hence can be scheduled.

\begin{theorem}
  \label{thm:flowtime}
  There exists an $O(1/\eps)$-approximation algorithm for
  non-clairvoyant DAG scheduling to minimize the weighted flow time on
  $m$ parallel machines, when there is a speedup of $2+\eps$.
\end{theorem}

The rest of this section gives the proof of Theorem~\ref{thm:flowtime}.
The algorithm remains unchanged from \S\ref{sec:compleTime} (we do not need
the algorithm~$\B$ now): we write
the convex program~(\ref{cp}) as before, which assign rates $\bL^t_j$ to
each job $j \in I_t$. The analysis again proceeds by writing a linear
programming relaxation, and showing a feasible dual solution. The LP is
almost the same as~(\ref{eq:LP}), just the objective is now (with
changes in red):
\[ \sum_{j,t} w_j \cdot \frac{\redd{(t - r_j)} \cdot x_{j,t}}{p_j}. \] Hence,
the dual is also almost the same as~(\ref{eq:dp}): the new dual
constraint requires that for every job $j$ and time $t \geq r_j$:
\begin{align}
  \label{eq:dualf}
  \alpha_j + \sum_{s \geq t} \left( \gout_{s,j} - \gin_{s,j} \right) \leq \beta_t \cdot p_j + w_j\redd{(t-r_j)}.
\end{align}
%Above, the color highlights the differences.

%--------------------------------------------------------------------------------
\subsection{Defining the Dual Variables}

In order to set the dual variables, define a total order $\prec$ on the
jobs as follows: First arrange the DAGs in order of release dates,
breaking ties arbitrarily. Let this order be $D_1, D_2, \ldots,
D_\ell$. All jobs in $D_i$ appear before those in $D_{i+1}$ in the order
$\prec$. Now for each dag $D_i$, arrange its jobs in the order they
complete processing by our algorithm. Note that this order is consistent
with the partial order given by the DAG. This also ensures that at any
time $t$, the set of waiting jobs in any DAG $D_i$ form a suffix in this
total order (restricted to $D_i$).

For a time $t$ and $j \in J_t$, let $\ind{j \in \Jact_t}$ denote the
indicator variable which is 1 exactly if $j$ is active at time $t$. The dual
variables are defined as follows:
\begin{itemize}
\item For a job $j \in J_t$, we set
  $\alpha_j := \sum_{t = r_j}^{C_j} \alpha_{j,t}$, where the 
  quantity $\alpha_{j,t}$ as defined as:
  $$ \alpha_{j,t} := \frac{1}{m} \Big[ w_j \cdot \ind{j \in \Jact_t} 
  \cdot \Big( \sum_{j' \in J_t: j' \preceq j} \btR_{j'} \Big) 
  + \btR_j \cdot \Big( \sum_{j' \in \Jact_t: j' \prec j} w_{j'} \Big)
  \Big]. $$
  
\item The variable $\beta_t := \frac{w(J_t)}{(1+\eps)m}$. Recall that the
  machines are allowed $2(1+\eps)$-speedup.

\item The definition of the $\gamma$ variables changes as follows. Let
  $(j' \rightarrow j)$ be an edge in the DAG $G_t$. Earlier we had
  considered paths $P_e$ containing $(j' \rightarrow j)$ only for the
  active edges $e$. But now we include {\em all} edges. Moreover, we
  replace the multiplier $\etat$ by $\etat_j$, where
  $ \etat_j := \frac1m \cdot  \Big( \sum_{j' \in J_t: j' \preceq j} w_{j'} \Big). $
  In other words, we define
  $$\gamma_{t, j' \rightarrow j} := \etat_j \cdot \sum_{e: e \in H_t, (j' \rightarrow j) \in P_e} \btz_e.$$
\end{itemize}
In the following sections, we show that these dual settings are enough
to ``pay for'' the flow time of our solution (i.e., have large objective
function value), and also give a feasible lower bound (i.e., are
feasible for the dual linear program). 

%--------------------------------------------------------------------------------
\subsection{The Dual Objective Function}
\label{sec:dual-object}

We first show that $\sum_j \alpha_j - m \sum_t \beta_t$ is close to the
total weighted flow-time of the jobs. The quantity $\chain_j$ is defined
as before. Notice that $\chain_j$ is still a lower bound on the
flow-time of job $j$ in the optimal schedule because all jobs of a DAG
are simultaneously released.
The following claim, 
whose result is deferred to the appendix, shows that the dual objective value is 
close to the weighted flow time of the algorithm. 

\begin{claim}
\label{cl:objf}
The total weighted flow-time  is at most
$ \ts \frac{2}{\eps} \Big(\sum_{j} \alpha_j
- m \sum_t \beta_t + \sum_j w_j \cdot \chain_j \Big). $
\end{claim}

\subsection{Checking Dual Feasibility}
\label{sec:dual-feas}

Now we need to check the feasibility of the dual
constraint~(\ref{eq:dualf}). In fact, we will show the following weaker
version of that constraint:
\begin{align}
\label{eq:dualf0}
\alpha_j + \redd{2} \sum_{s \geq t} \left( \gout_{s,j} - \gin_{s,j} \right) \leq \beta_t \cdot p_j + \redd{2}w_j(t-r_j).
\end{align}
This suffices to within another factor of $2$: indeed, scaling down the
$\alpha$ and $\beta$ variables by another factor of $2$ then gives dual
feasibility, and loses only another factor of $2$ in the objective function.
We begin by bounding $\alpha_{j,s}$ in two different ways.

\begin{lemma}
  \label{lem:af}
  For any time $s \geq r_j$, we have $\alpha_{j,s} \leq 2 w_j$. 
  % Further, if $j$ is inactive at time $s$, then $\alpha_{j,s} \leq \eta^s \cdot \bR^s_j. $
\end{lemma}
\begin{proof}
  Consider the second term in the definition of $\alpha_{j,s}$. This
  term contains $\sum_{j' \in \Jact_s: j' \prec j} w_{j'}$. By
  Corollary~\ref{cor:sum},  for any $j' \in \Jact_s$ we have 
  $w_{j'} = \bR^s_{j'} \cdot \eta^s$. Therefore,
  $$ \sum_{j' \in \Jact_s: j' \prec j} w_{j'} \quad \leq \quad \eta^s \cdot \sum_{j'
    \in \Jact_s: j' \prec j} \bR^s_{j'}\quad \leq \quad \eta^s \cdot \sum_{j' \in
    J_s} \bR^s_{j'}. $$ 
  Now we can bound  $\alpha_{j,s}$ by dropping the indicator on the
  first term to get
  \begin{align*}
    \frac{1}{m} \cdot \Big[ \Big(w_j  \cdot \sum_{j' \in J_s: j' \preceq j} \bR^s_{j'} \Big) 
    + \bR^s_j \cdot \Big( \eta^s \cdot \sum_{j' \in \Jact_s: j' \prec j} \bR^s_{j'}  \Big) \Big] 
    &~ \leq~  \frac{1}{m}  w_j \Big[ \sum_{j' \in J_s} \bR^s_{j'} + \sum_{j' \in J_s} \bR^s_{j'} \Big] ,
  \end{align*}
  the last inequality using Claim~\ref{cl:sum0}. Simplifying,
  $\alpha_{j,s} \leq  \frac{2}{m} \cdot w_j \cdot \sum_{j'' \in I_s} \bL^s_{j''} =2 w_j$.
\end{proof}

Here is a slightly different upper bound on $\alpha_{j,s}$.
\begin{lemma}
  \label{lem:afnew}
  For any time $s \geq r_j$, we have
  $\alpha_{j,s} \leq 2 \eta^s_j \cdot \bR^s_j $.
%Further, if $j$ is inactive at time $s$, then $\alpha_{j,s} \leq \eta^s \cdot \bR^s_j. $
\end{lemma}

\begin{proof}
  The second term in the definition of $\alpha_{j,s}$ is at most
  $\eta^s_j \cdot \bR^s_j$, directly using the definition of
  $\eta^s_j$. For the first term, assume $j$ is active at time $s$,
  otherwise this term is 0. Now Corollary~\ref{cor:sum} shows that
  $w_j = \eta^s \cdot \bR^s_j$, so the first term can be bounded as follows:
  \[
    \frac{w_j}{m} \cdot \sum_{j' \in J_s: j' \preceq j} \bR^s_{j'} \quad = \quad
    \frac{\bR^s_j \cdot \eta^s}{m} \cdot \sum_{j' \in J_s: j' \preceq j}
    \bR^s_{j'} \quad \stackrel{\small \text{(Claim~\ref{cl:sum0})}}{\leq} \quad \frac{\bR^s_j}{m} \cdot \sum_{j' \in J_s: j' \preceq
      j} w_{j'} \quad = \quad \bR^s_j \cdot \eta^s_j, \] which completes the proof.
\end{proof}

To prove~\eqref{eq:dualf0}, we write
$\alpha_j = \sum_{s = r_j}^{t-1} \alpha_{j,s} + \sum_{s \geq t}
\alpha_{j,s}$, and use Lemma~\ref{lem:af} to cancel the first summation
with the term $2w_j(t-r_j)$. Hence, it remains to prove
\begin{align}
\label{eq:dualf1}
\sum_{s \geq t} \alpha_{j,s} + 2 \sum_{s \geq t} \left( \gout_{s,j} - \gin_{s,j} \right) \leq \beta_t \cdot p_j. 
\end{align}

Let $t_j^\star$ be the time at which the algorithm starts processing $j$. We first argue why we can ignore times $s< t_j^\star$ on the LHS of \eqref{eq:dualf1}.
%We first consider a time $s$ which is at most $t^\star_j$. 
\begin{claim}
  \label{cl:sf}
  Let $s$ be a time satisfying $r_j \leq s < t_j^\star$. Then $\alpha_{j,s} + 2(\gout_{s,j} - \gin_{s,j}) \leq 0. $
\end{claim}
\begin{proof}
  While computing $\gout_{s,j} - \gin_{s,j}$, we only need to consider
  paths $P_e$ for edges $e$ in $H_s$ which have $j$ as end-point. Since
  $j$ does not appear on the left side of $H_s$, this quantity is equal
  to $-\eta^s_j \cdot \bR^s_j$. The result now follows from
  Lemma~\ref{lem:afnew}.
\end{proof}

So using Claim~\ref{cl:sf} in~(\ref{eq:dualf1}), it suffices to show
\begin{align}
\label{eq:dualf2}
  \sum_{s \geq \max\{t,t_j^\star\}} \alpha_{j,s} + 2 \sum_{s \geq
  \max\{t,t_j^\star\} } \left( \gout_{s,j} - \gin_{s,j} \right) \leq
  \beta_t \cdot p_j.  
\end{align}
Note that we still have $\beta_t$ on the right hand side, even though
the summation on the left is over times $s \geq \max\{t,t_j^\star\}$.
The proof of the following claim is deferred to appendix.

\begin{claim}
  \label{cl:sf1}
  Let $s$ be a time satisfying $s \geq \max\{t,t_j^\star\}$. Then
  $\alpha_{j,s} + 2(\gout_{s,j} - \gin_{s,j}) \leq 2(1+\eps) \beta_t
  \cdot \bL^s_j. $
\end{claim}

Hence, the left-hand side of~\eqref{eq:dualf2} is at most
$2(1+\eps) \beta_t \cdot \sum_{s \geq  \max\{t,t_j^\star\} } \bL^s_j$. However,
since job $j$ is assigned a rate of $\bL^s_j$ and the machines run at
speed $2(1+\eps)$, we get that this expression is at most
$p_j \cdot \beta_t$, which is the right-hand side
of~(\ref{eq:dualf2}). This proves the feasibility of the dual
constraint~(\ref{eq:dualf0}).

%--------------------------------------------------------------------------------
%\subsection{Wrapping Up}

\begin{proof}[Proof of Theorem~\ref{thm:flowtime}]
  In the preceding \S\ref{sec:dual-feas} we proved that the
  variables $\alpha_j/2$, $\beta_t/2$ and $\gamma_{t,j'\to j}$ satisfy
  the dual constraint for the flow-time relaxation.
  Since $\sum_j (\alpha_j/2) - m \sum_t (\beta_t/2)$ is a feasible dual,
  it gives a lower bound on the cost of the optimal solution. Moreover,
  $\sum_j w_j \cdot \chain_j$ is another lower bound on the cost of the
  optimal schedule. Now using the bound on the weighted flow-time of our
  schedule given by Claim~\ref{cl:objf}, this shows that we have an
  $O(1/\eps)$-approximation with $2(1+\eps)$-speedup.
\end{proof}

In \S\ref{sec:improve-flowtime} we show how to use a slightly different scheduling
policy that prioritizes the last arriving jobs to reduce the speedup to
$(1+\eps)$. 

%--------------------------------------------------------------------------------

%%% Local Variables:
%%% mode: latex
%%% TeX-master: "main"
%%% End:

% !TeX root = main.tex 
% !TEX root = main.tex

%--------------------------------------------------------------------------------
\section{An $O(1/\eps^2)$-competitive Algorithm with $(1+\eps)$-speed}
\label{sec:improve-flowtime}

Theorem~\ref{thm:flowtime} requires $(2+\eps)$ speedup. In this section, we improve the speed scaling requirement to $(1+\eps)$. We prove the following:

\begin{theorem} \label{thm:flowtimenew}
There exists an $O(1/\eps^2)$-approximation algorithm for non-clairvoyant DAG scheduling to minimize weighted flow time on parallel machines when there is a speedup of $1+\eps$.
\end{theorem}
For ease of exposition, we assume a $(1+3\eps)$-speedup in the proof of Theorem~\ref{thm:flowtimenew}.
%--------------------------------------------------------------------------------
\subsection{The Algorithm}
The algorithm remains unchanged -- we shall assign rates $\btL_j$ to each job $j \in I_t$. These rates are derived by a suitable convex program. This convex program is again same as~\eqref{cp}, except that the objective function now changes to 
$$\sum_{j \in J_t} \redd{\wh_{j,t}} \ln \btR_j, $$
where we replace the weight $w_j$ of job $j$ by a new time dependent quantity $\wh_{j,t}$  defined as follows.  

\begin{definition}[Weight $\wh_{j,t} $]
Consider a time $t$, and let $J_{<j,t}$ denote the set of jobs in $J_t$ which appear before $j$ in the ordering $\prec$. Define $J_{\leq j, t}$ similarly (it includes $j$ as well). Let $k$ denote $1/\eps$. We define $$ \wh_{j,t} := \frac{w(J_{\leq j, t})^k - w(J_{< j,t})^k}{w(J_t)^k}. $$
\end{definition}

It is easy to check that $\sum_{j \in J_t} \wh_{j,t} = 1$. Moreover, since $f(x) = x^k$ is a convex function, we have the following easy fact.
\begin{fact} \label{fact:convexity}
We have 
\[ k w_j \cdot \frac{ w(J_{< j,t})^{k-1}}{w(J_t)^k} \quad \leq \quad 
\wh_{j,t}  \quad \leq \quad k w_j \cdot \frac{ w(J_{\leq j,t})^{k-1}}{w(J_t)^k} .
\]
\end{fact}

This completes the description of the algorithm. 

%--------------------------------------------------------------------------------
\subsection{The Convex Program and Nice Times}

We now briefly indicate how the analysis of the algorithm  gets adapted to this algorithm. 
The KKT condition~\eqref{k1} now changes to 
\begin{align}
\label{k1new}
\frac{\wh_{j,t}}{\btR_j} & = \theta^t_{j'} + \eta^t - \nu_e^t \qquad  \forall e=(j',j), j' \in I_t, j \in J_t .
\end{align}

The KKT conditions~\eqref{k4}--\eqref{k6} remain unchanged. Hence,
Claim~\ref{cl:sum0} and Corollary~\ref{cor:sum} get restated thus:
\begin{claim}
\label{cl:sumnew}
Consider a job $j \in J_t$. Then $\wh_{j,t} \geq \btR_j \cdot \etat. $ Further, if $j$ is active at time $t$, then $\wh_{j,t} = \btR_j \cdot \etat$. 
\end{claim}

We now introduce a useful definition. 
\begin{definition} [Nice time]
We say that a time $t$ is {\em nice} if $w(\Jact_t) \geq (1-\eps) \cdot w(J_t)$.  
\end{definition}
Let $\Tnc$ denote the set of nice time slots. Claim~\ref{cl:upper} can
now be restated as:
\begin{claim} 
  \label{cl:upperLAPS} For any time $t$, we have
  $\wh(\Jact_t)/m \leq \etat \leq \wh(J_t)/m$. Further, if $t \in \Tnc$, then  $1/e \leq \etat \cdot m \leq 1. $
\end{claim}
\begin{proof}
The first statement follows as in Claim~\ref{cl:upper}. So, it remains to prove the second claim. Again, $m \cdot \etat \leq 1$ follows
from the fact that $\wh(J_t)=1$ (by definition). Now, let us estimate $\wh(\Jact_t)$. Again by definition of $\wh$, it is easy to see that 
\[ \wh(\Jact_t) \quad \geq  \quad \frac{w(\Jact_t)^k}{w(J_t)^k} \quad \geq \quad (1-\eps)^k \quad \geq \quad 1/e.   \qedhere \]
\end{proof}

The definitions of the dual variables $\alpha, \beta, \gamma$ get slightly modified. 
The quantity  $\alpha_{j,s}$ is non-zero only when $s$ is nice. In other words, 
\begin{align} \label{def:alphaScalable}
\alpha_{j,s} := \frac{\ind{s \in \Tnc}}{m} \Big[ w_j \cdot \ind{j \in \Jact_s} 
  \cdot \Big( \sum_{j' \in J_s: j' \preceq j} \btR_{j'} \Big) 
  + \btR_j \cdot \Big( \sum_{j' \in \Jact_s: j' \prec j} w_{j'} \Big)
  \Big].
  \end{align}
The dual variables $\beta_t$ and $\gamma_e$ are defined as before. 
We first show the analogue of Claim~\ref{cl:objf}. 
\begin{claim}
\label{cl:objflaps}
The total weighted flow-time of the jobs is at most
\[ \frac{2}{\eps} \Big(\sum_{j} \alpha_j
- m \cdot \sum_t \beta_t \Big) + \frac{2}{\eps^2} \cdot \sum_j w_j \cdot \chain_j . \]
\end{claim}
\begin{proof}
Consider a nice time  $t \in \Tnc$. As in the proof of Claim~\ref{cl:objf}, we get 
\[ \sum_{j \in J_t} \alpha_{j,t} \quad = \quad w(\Jact_t) \quad \geq \quad (1-\eps) \cdot w(J_t),\]
where the last inequality follows from the fact that $t$ is nice. The following inequality follows as in the 
proof of Claim~\ref{cl:objf} (note that the machines run at speed $(1+3\eps)$ now). 
\[ \sum_t w(J_t \setminus \Jact_t) \quad \leq \quad \sum_j w_j \cdot \frac{\chain_j}{1+3 \eps}  \quad \leq \quad \sum_j w_j \cdot {\chain_j}. \]
Now consider a $t \notin \Tnc$. This means $w(\Jact_t) \leq (1-\eps) \cdot w(J_t)$, or $w(J_t) \leq \frac{1}{\eps} w(J_t \setminus \Jact_t) $. Thus,
$$ \sum_{t \notin \Tnc} w(J_t) \quad \leq \quad \frac{1}{\eps} \cdot \sum_{t \notin \Tnc} w(J_t \setminus \Jact_t) 
\quad \leq \quad \frac{1}{\eps} \cdot \sum_j w_j \cdot \chain_j. $$
This means
$$ \sum_{t} w(J_t) \quad = \quad \sum_{t \in \Tnc} w(J_t)   + \sum_{t \notin \Tnc} w(J_t) \quad \leq \quad \frac{1}{1-\eps} \cdot \sum_{j} \alpha_j + \frac{1}{\eps} \cdot \sum_j w_j \cdot \chain_j. $$
Since $m \cdot \sum_t \beta_t = \sum_t w(J_t)/(1+\eps)$, taking difference we get
\[   \sum_t w(J_t) \cdot \Big(1- \frac{1}{1+\eps} \Big) \quad \leq \quad  \frac{1}{1-\eps} \cdot \sum_{j} \alpha_j + \frac{1}{\eps} \cdot \sum_j w_j \cdot \chain_j - m \cdot \sum_t \beta_t, 
\]
which implies the claim because the total weighted flow-time equals $\sum_t w(J_t)$.
\end{proof}

%--------------------------------------------------------------------------------
\subsection{Checking Dual Feasibility} \label{sec:dualFeasibScalable}

We now want to check the dual constraint~(\ref{eq:dualf}), so fix
a job $j$.
Lemmas~\ref{lem:af}
and~\ref{lem:afnew} get modified as follows. 

\begin{lemma}
\label{lem:aflaps}
For any time $s \geq r_j$, we have $\alpha_{j,s} \leq ke \cdot w_j$. 
\end{lemma}
\begin{proof} 
We can assume that $s$ is nice, otherwise $\alpha_{j,s}$ is 0. 
  Consider the first term in the definition of definition of $\alpha_{j,s}$. 
  Since 
$  \sum_{j' \in J_s: j' \preceq j} \bR^s_{j'}  \leq  \sum_{j' \in J_s} \bR^s_{j'}  =   \sum_{j' \in J_s} \bL^s_{j'}  \leq m$,
  this term  
  \[ \frac{1}{m} \Big[ w_j \cdot \ind{j \in \Jact_s} 
  \cdot \Big( \sum_{j' \in J_s: j' \preceq j} \bR^s_{j'} \Big) \Big] \leq w_j.
  \]

  Now consider the second term  of $\alpha_{j,s}$.
By Claim~\ref{cl:sumnew}
  we have $\bR^s_j  \leq \frac{\wh_{j,s}}{\etat}$, which implies
  \begin{align*}
  \frac1m \bR^s_j \cdot \Big( \sum_{j' \in \Jact_s: j' \prec j} w_{j'} \Big)  &\leq  \frac{\wh_{j,s}}{m\eta^s} \cdot \Big( \sum_{j' \in \Jact_s: j' \prec j} w_{j'} \Big) .
  \end{align*}
  Now using Fact~\ref{fact:convexity},
  \begin{align*} 
    \frac1m \bR^s_j \cdot \Big( \sum_{j' \in \Jact_s: j' \prec j} w_{j'} \Big) 
  \quad  \leq \quad \frac{1}{m\eta^s} k w_j \cdot \frac{ w(J_{\leq j,s})^{k-1}}{w(J_s)^k} \cdot w(\Jact_s) \quad \leq \quad \frac{k \cdot w_j}{m\eta^s} \quad \leq \quad k e \cdot w_j, %\cdot \frac{ w(J_{\leq j,t})^{k}}{w(J_t)^k}  
  \end{align*}
where the last inequality follows from Claim~\ref{cl:upperLAPS}. 
\end{proof}

\begin{lemma}
\label{lem:afnewlaps}
For any time $s \geq r_j$, we have $\alpha_{j,s} \leq (1+\eps) \cdot  \eta^s_j \cdot \bR^s_j. $
\end{lemma}
\begin{proof}
The second term in definition of $\alpha_{j,s}$ from \eqref{def:alphaScalable} is easy to bound because
\[ \frac{\btR_j}{m} \cdot \sum_{j' \in \Jact_s: j' \prec j} w_{j'} \leq \bR_j \cdot \eta^s_j \]
by the definition of $\eta^s_j$. It remains to bound the first term. Assume that $j$ is active. By Claim~\ref{cl:sumnew} and the definition of $\wh^{s}_{j'}$, we get 
$$ \frac{w_j}{m} \cdot  \sum_{j' \in J_s, j' \prec j} \bR^s_{j'} \quad = \quad \frac{w_j}{m \cdot \eta^s} \cdot \sum_{j' \in J_s, j' \prec j} \wh^s_{j'} 
\quad= \quad \frac{w_j}{m \cdot \eta^s} \cdot \frac{w(J_{< j,s})^k}{w(J_s)^k}.
$$
Using Fact~\ref{fact:convexity}, the above can be upper bounded by 
$$ \frac{\wh_{j,s} \cdot w(J_{<j,s})}{k \cdot m \cdot \eta^s} = \frac{\bR^s_j \cdot \eta^s_j}{k},$$
where the last term follows from the fact that $j$ is active. This proves the claim because
\[ \alpha_{j,s} \quad \leq \quad \frac{\btR_j}{m} \cdot \sum_{j' \in \Jact_s: j' \prec j} w_{j'}  + \frac{w_j}{m} \cdot  \sum_{j' \in J_s, j' \prec j} \bR^s_{j'} \quad \leq \quad  \bR_j \cdot \eta^s_j + \frac{\bR^s_j \cdot \eta^s_j}{k}. \qedhere
\]
\end{proof}

The rest of the arguments follow as in the previous section. We can show
in a similar manner that for any job $j$ and time $t$:
\begin{align}\label{eq:dualf0New}
\alpha_j +  \redd{(1+\eps)} \cdot \sum_{s \geq t} \left( \gout_{s,j} - \gin_{s,j} \right) \leq \beta_t \cdot  p_j  + 
\redd{ke} \cdot w_j \cdot (t-r_j). 
\end{align}
This suffices because it implies that $\frac{\alpha_j}{ke}, \frac{\beta}{ke},$ and  $(1+\eps) \frac{\gamma}{ke}$ are feasible dual solutions, which loses only another factor of $ke$ in the objective function $\sum_j \alpha_j - m\sum_t \beta_t$.

We first argue using Lemma~\ref{lem:aflaps} that it suffices to show
\[  \sum_{s\geq t} \alpha_{j,s} +  {(1+\eps)} \cdot \sum_{s \geq t} \left( \gout_{s,j} - \gin_{s,j} \right) \leq \beta_t \cdot  p_j ,
\]
and then further simplify it to showing
\begin{align} \label{eq:dualf2New}
  \sum_{s\geq \max\{t,t_j^{\star}\}} \alpha_{j,s} +  {(1+\eps)} \cdot \sum_{s\geq \max\{t,t_j^{\star}\}} \left( \gout_{s,j} - \gin_{s,j} \right) \leq \beta_t \cdot  p_j 
\end{align}
because for any time $s$  satisfying $r_j \leq s < t_j^\star$, a variant of Claim~\ref{cl:sf} shows $\alpha_{j,s} + (1+\epsilon)(\gout_{s,j} - \gin_{s,j}) \leq 0$. Here we get a factor $(1+\eps)$ instead of factor $2$ in Claim~\ref{cl:sf}  because Lemma~\ref{lem:afnewlaps} has $(1+\eps)$ factor unlike Lemma~\ref{lem:afnew}.
Finally, Claim~\ref{cl:sf1} now gets modified as follows; we omit the proof since it is essentially unchanged.
\begin{claim}
  Let $s$ be a time satisfying $s \geq \max\{t,t_j^\star\}$. Then
  $\alpha_{j,s} + (1+\eps )(\gout_{s,j} - \gin_{s,j}) \leq (1+3\eps) \beta_t
  \cdot \bL^s_j. $
\end{claim}

% \begin{proof} %\snote{We can comment this proof.}
%   We begin by bounding $(1+\eps)(\gout_{s,j} - \gin_{s,j})$. The
%    contribution from paths $P_e$ for which $j$
%   lies on the right side of $H_s$ is $-(1+\eps) \eta^s_j \bR^s_j$, which by
%   Lemma~\ref{lem:afnewlaps} cancels $\alpha_{j,s}$. Thus we get
%   $$ \alpha_{j,s} + (1+\eps)(\gout_{s,j} - \gin_{s,j}) = (1+\eps) \sum_{j' \in J_s :
%     e=(j \rightarrow j') \in E_s} \eta^s_{j'} \cdot \bz^s_{e}. $$
%   Since all jobs in a DAG have the same release
%   time, any job $j'$ in the summation above is released at the
%   same time as $j$. Moreover, any job $j'' \in J_s$ which contributes
%   towards
%   $\eta^s_{j'} = \frac1m \cdot ( \sum_{j'' \in J_s: j'' \preceq j'}
%   w_{j''} )$ has  been  also released at or before $r_j$. Therefore,
%   $\eta^s_{j'} \leq w(J_t)/m = (1+\eps)\beta_t$. 
%   This implies 
%   \[ \alpha_{j,s} + (1+\eps)(\gout_{s,j} - \gin_{s,j}) \quad \leq \quad (1+\eps) \sum_{j' \in J_s :
%     e=(j \rightarrow j') \in E_s} (1+\eps)\beta_t \cdot \bz^s_{e} \quad \leq \quad
%      (1+3\eps) \beta_t
%   \cdot \bL^s_j
%   \]
%   using $\bL^s_j = \sum_{j' \in J_s :
%     e=(j \rightarrow j') \in E_s} \bz^s_{e}$.
% \end{proof}

Hence, the left-hand side of~\eqref{eq:dualf2New} is at most
$(1+3\eps) \beta_t \cdot \sum_{s \geq  \max\{t,t_j^\star\} } \bL^s_j$. However,
since job $j$ is assigned a rate of $\bL^s_j$ and the machines run at
speed $(1+3\eps)$, we get that this expression is at most
$p_j \cdot \beta_t$, which is the right-hand side
of~(\ref{eq:dualf2New}). This proves the feasibility of the dual
constraint~(\ref{eq:dualf0New}).

%--------------------------------------------------------------------------------
\subsection{Wrapping Up}

\begin{proof}[Proof of Theorem~\ref{thm:flowtimenew}]
  In the preceding \S\ref{sec:dualFeasibScalable} we proved that the
  variables $\frac{\alpha_j}{ke}, \frac{\beta}{ke},$ and  $(1+\eps) \frac{\gamma}{ke}$
   satisfy  the dual constraint for the flow-time relaxation.

  Since $\sum_j \big(\alpha_j/(ke)) - m \sum_t (\beta_t/(ke)\big)$ is a feasible dual,
  it gives a lower bound on the cost of the optimal solution. Moreover,
  $\sum_j w_j \cdot \chain_j$ is another lower bound on the cost of the
  optimal schedule. Now using the bound on the weighted flow-time of our
  schedule given by Claim~\ref{cl:objflaps}, this shows that we have an
  $O(1/\eps^2)$-approximation with $(1+3\eps)$-speedup.
\end{proof}

%%% Local Variables:
%%% mode: latex
%%% TeX-master: "main"
%%% End:

% !TeX root = main.tex 
% !TEX root = main.tex

\section{Lower Bounds}
\label{sec:lower-bounds}

For the problem of minimizing weighted completion time under precedence
constraints, we allow the jobs in the DAG to arrive over time, and hence
different jobs can have different release dates. (All we require is that
the release dates respect the order given by the DAG, so a job with an
earlier release date cannot depend on a job with a later one.)
However, in the case of weighted \emph{flow-time} minimization, we
insist that  jobs in the same DAG have the same release date. We now show that
this assumption is necessary: if we allows jobs in a DAG to arrive over
time, there are strong lower bounds even for a single machine and in the
clairvoyant setting (i.e., when the algorithm knows the size of a job
when it arrives).

\begin{theorem}[Lower Bound]
  \label{thm:lower}
  Any randomized online algorithm for the problem of minimizing
  unweighted flow-time on a single machine with precedence constraints
  and release dates has an unbounded (expected) competitive ratio even
  in the clairvoyant setting. This lower bound holds even if we allow
  the speed of the machine to be augmented by a factor of $c$, for any
  constant $c > 0$.
\end{theorem}

\begin{proof}
  We give a probability distribution over inputs, and show that the
  expected competitive ratio of any deterministic algorithm is
  unbounded. By Yao's Lemma, this implies the desired lower bound.

  Initially, $n$ jobs arrive at time $0$, each of them has size 1. At
  time $1$, we choose one of these jobs uniformly at random, say
  $j \in [n]$, and release $n^3$ new jobs where each new job $j'$ depends
   on $j$, i.e., $j \prec j'$.    Hence the
  precedence graph is a star with $n^3$ leaves, rooted at $j$, along
  with the items in $[n]\setminus \{j\}$ which are unrelated to elements
  of this star. These  $n^3$ new jobs have 0 size. The parameter $n$ is
  assumed to be much larger than the speedup $c$.

  Let us first consider the offline optimum. It schedules the job $j$
  in the interval $[0,1]$ and so completes it---the $n^3$ jobs arriving
  at time $1$ can now be finished immediately, and hence the flow-time
  for them is zero. It finally schedules the remaining $n-1$ jobs of size $1$
   that had arrived at time 0. Their total flow-time is $O(n^2)$.

  Now consider any deterministic online algorithm. By time $1$, it can
  perform $c \ll n$ amounts of processing, and so at least half the jobs
  will have seen less than $1/2$ amount of processing. The randomly
  chosen job $j$ is such a job with probability at least $1/2$. If this
  event happens, the flow-time of the arriving $n^3$ jobs would be at
  least $n^3/2$, and hence the expected flow-time of this algorithm is
  $\Omega(n^3)$.
\end{proof}

This shows why we need our assumption that the release times of any two
related jobs is the same. This is a reasonable assumption for many
settings, e.g., in~\cite{RS,ALLM} where each job is a DAG of
tasks. We extend their model  from minimizing unweighted
flow-time of jobs to weighted flow-time of tasks. 

%%% Local Variables:
%%% mode: latex
%%% TeX-master: "main"
%%% End:

% !TeX root = main.tex 
% !TEX root = main.tex

\newcommand{\Hp}{H^{(p)}}
\newcommand{\Ip}{I^{(p)}}
\newcommand{\Jp}{J^{(p)}}
\newcommand{\Dp}{{(\Delta_p)}}

\section{Solving the Convex Program}
\label{sec:convsolv}

Our results in the previous
sections % \S\ref{sec:compleTime} and \S\ref{sec:flow}
rely on solving the convex program~\eqref{cp} to assign rates to the
minimal jobs.  In this section we show that we do not need a generic
convex program solver for this purpose: we can run an efficient
``water-filling'' algorithm instead. Indeed, combinatorial algorithms to
solve the Eisenberg–Gale convex program (and other problems in market
equilibria) have been studied widely, starting with the work of Devanur
et al.~\cite{DPSV-JACM08}. Specifically, the constraints~\eqref{eq:c1},
\eqref{eq:c2}, \eqref{eq:c3}, and \eqref{eq:c5} in~\eqref{cp} are a
special case of the Eisenberg–Gale convex program for linear Fisher
markets when the utility derived from different goods is the same. On
one hand, this means our setting is easier and we can use water-filling
to solve the program (whereas such a simple algorithm does not suffice
with general utilities~\cite{DPSV-JACM08}). On the other hand it does
not seem possible to use the prior results directly, since we have an 
additional global constraint~\eqref{eq:c4} in~\eqref{cp}.

Since this convex program is solved once at every time $t$
during the online algorithm, we consider a fixed time $t$ and remove all
subscripts involving $t$ in this section.
We have a bipartite graph $H$ with the left side being $I$ and the right side denoted by $J$. We shall use $E$ to denote the set of edges here. Every vertex $j \in I$ has an associated variable $L_j$ and the vertices $j \in J$ have variables $R_j$ associated with them. Further  we have a variable $z_e$ for every edge $e \in E$. For a subset $J'$ of $J$, define $\Gamma(J')$ as its set of neighbors in $I$. For a vertex $v$, define $\delta(v)$ to be the set of edges incident to it. 
 There is a notion of time  in our algorithm that increases at a uniform rate. We  use $T$ to denote this time variable. Our algorithm maintains a feasible solution  at all times $T$. 

The idea of the algorithm is to proceed in phases, and to simultaneously increase all $R_j$ values (initialized at $0$) at rate $w_j$ while maintaining feasibility. A phase ends when the algorithm can no longer perform this increase. This could be because  of two reasons: (i) there is a tight set $J' \subseteq J$ with $|\Gamma(J')| = w(J')\cdot T$ or (ii) the constraint $\sum_{j\in J} R_j \leq m$ is tight. In the former case we make progress by removing sets $J'$ and $\Gamma(J')$, and in the latter case we finish with an optimal solution to~\eqref{cp}.

Formally, in a phase $p$ we shall consider a  sub-graph $\Hp$ of $H$. The left and the right sides of $\Hp$ are denoted $\Ip$ and $\Jp$, respectively. In fact, $\Hp$ is the subgraph of $H$ induced by $\Ip$ and $\Jp$, and so, it will suffice to specify the latter two sets. The algorithm is described in Algorithm~\ref{alg:cp}. Although in this description we raise $T$ (and hence $R_j$) continuously, this can be implemented in polynomial time using parametric-flows~\cite{GGT-SICOMP89}.
We now argue the algorithm's correctness (i.e., it outputs  a feasible solution)  and then prove its optimality. 

%-------------------------------------------------------
\subsection{Correctness}
In order to prove correctness we need to show that the fractional assignments mentioned in Steps~\ref{m1} and~\ref{m2} can always be found. 
 We first show the algorithm always maintains  $L_j \leq 1$ for all $j \in I$. In Claim~\ref{cl:m2} we  argue that $\sum_{j \in I} L_j \leq m$, which implies feasibility for~\eqref{cp}.

We show that the following invariant is always maintained at any time $T$ during the algorithm. 

\begin{claim}
\label{cl:Hall}
Consider a time $T$ during a phase $p$ of the algorithm. There exist non-negative values $z_e$ for all edges $e$ in the graph $H_p$ such that the following conditions are satisfied:
\begin{itemize}
\item For every $j \in \Jp$, we have $\sum_{e \in \delta(j)} z_e = w_j \cdot T. $
\item For every $j \in \Ip$, we have $\sum_{e \in \delta(j)} z_e \leq 1. $
\end{itemize}
\end{claim}
\begin{proof}
We prove the following statement by induction on phase $p$: at any time $T$ during a phase $p$ of the algorithm, $w(J') \cdot T \leq |I'|$ for every subset $J' \subseteq \Jp$ and $I' = \Gamma(J')$. It is easy to see that once we show this statement, the desired result follows by Hall's matching theorem.

It is clearly true for  $p=0$. Suppose it is true for some time $T = T_1$ in phase $p$, and we increase $T$ from $T_1$ to $T_2$ during this phase. Consider a subset $J'$ of $\Jp$, and let $I'$ denote $\Gamma(J')$. By induction hypothesis, $w(J') \cdot T_1 \leq |I'|$. As we raise $T$, the LHS will increase but the RHS remains unchanged. If the two become equal, this phase will end. Since $T_2$ also lies in this phase, $w(J') \cdot T_2$ must be at most $|I'|$, and the invariant continues to hold at time $T_2$.

 Now suppose we go from phase $p$ to phase $p+1$ at time $T$. Let $I', J'$ be as defined in Step~\ref{i}. 
 Suppose this invariant is violated at time $T$ in phase $p+1$, i.e., there exist subsets $J''$ and $I''= \Gamma(J'')$ of 
 $J^{(p+1)}$ and $I^{(p+1)}$, respectively, for which $w(J'') \cdot T > |I''|.$ Now consider the set of vertices
 $J' \cup J''$ in $\Hp$. Clearly $\Gamma(J' \cup J'') = I' \cup I''$. But then $w(J' \cup J'') \cdot T > |I'| + |I''| = |I' \cup I''|$, which contradicts the fact that the invariant condition always holds in phase $p$. 
\end{proof}

\begin{corollary}
The algorithm will find the desired matching is Steps~\ref{m1} and~\ref{m2}. 
\end{corollary}
\begin{proof}
Consider the assignment required in Step~\ref{m1}. Let $z$ be the assignment guaranteed by Claim~\ref{cl:Hall}, and consider its restriction to edges in $E'$. Since $I' = \Gamma(J')$, it follows that 
$$ \sum_{e \in E', e \in \delta(j)} z_e \quad = \quad \sum_{e \in \Hp, e \in \delta(j)} z_e  \quad = \quad R_j. $$
We also know that for any $j \in I'$, 
$$\sum_{e \in E', e \in \delta(j)} z_e \leq 1. $$
But note that $\sum_{j \in J'} R_j = \sum_{j \in I'} L_j$. The former quantity is equal to $w(J') \cdot T$, while the latter is at most $|I'|$. But we know from the condition in Step~\ref{i} that they are equal. Therefore $L_j = 1$ 
for all $j \in I'$. This yields the desired assignment for Step~\ref{m1}. 
The desired assignment for Step~\ref{m2} follows directly from Claim~\ref{cl:Hall}. 
\end{proof}

\begin{algorithm*}[h]
\caption{Solving the Convex Program~\eqref{cp}}
\label{alg:cp}
\begin{algorithmic}[1]
\State Initialize $T \leftarrow 0, p \leftarrow 0$. 
\State Initialize $\Hp \leftarrow H, \Ip \leftarrow I,
\Jp \leftarrow J$. 
\State Initialize the variables $z, L, R$ to 0. 
\Repeat{} \label{repeat}
\State Raise $T$ at a uniform rate till one of the following two events happen:
\Statex
\State (i) There is a subset $J' \subseteq \Jp$ for which the set $I' = \Gamma(J')$ has cardinality $w(J') \cdot T$. \label{i}
%%\State \quad \quad Let $J'$ be such a maximal subset of $\Jp$ and $I'$ denote $\Gamma(J').$ \label{IJ}
\State \qquad \quad  For every $j \in J'$, set $R_j \leftarrow w_j \cdot T$. 
\State \qquad \quad  For every $j \in I'$, set  $L_j \leftarrow 1$. 
\State \qquad \quad Let $E'$ be the set of edges between $I'$ and $J'$. 
\State \qquad \quad  For every edge $e \in E'$, set $z_e$ to values satisfying : 
$$ \sum_{e \in E', e \in \delta(j)} z_e = R_j, \ \forall j \in J'; \  \sum_{e \in E', e \in \delta(j)} z_e = L_j, \ \forall j \in I'. $$ \label{m1}
\State \qquad \quad  $J^{(p+1)} \leftarrow J^{(p)} \setminus J'$ and $I^{(p+1)} \leftarrow I^{(p)} \setminus I'$. 
\State \qquad\quad Terminate if $I^{(p+1)}= \emptyset$. \label{term(i)}
\State \qquad \quad  $p \leftarrow p+1$, Goto Step~\ref{repeat}. 
\Statex
\State (ii) $\sum_{j \in \Jp} w_j \cdot T + |I \setminus \Ip| = m. $ \label{ii}
\State \qquad \quad  For every $j \in \Jp$, set $R_j \leftarrow w_j \cdot T$. 
\State \qquad \quad  For every edge $e$ in $\Hp$, set $z_e$ to values satisfying : 
$$ \sum_{e \in \Hp, e \in \delta(j)} z_e = R_j, \ \forall j \in \Jp; \  L_j := \sum_{e \in \Hp, e \in \delta(j)} z_e \leq 1, \ \forall j \in \Ip. $$ \label{m2}
\State \qquad \quad Terminate. \label{term(ii)}
\Until{$T$ cannot be raised.}
\end{algorithmic}
%\end{center}
\end{algorithm*}

We now know the algorithm always ensures that $L_j \leq 1$ for all $j \in I$. Next we show that it maintains the invariant  $\sum_{j \in I} L_j \leq m$. This will show that these values are feasible for~\eqref{cp}. 

\begin{claim}
\label{cl:m2}
When the algorithm terminates, $\sum_{j \in I} L_j \leq m$. Further, if it terminates after executing Step~\ref{ii}, then
$\sum_{j \in I} L_j = m$. 
\end{claim}
\begin{proof}
We first show by induction on  phase $p$ that the following condition always holds for all $T$:
$$ \sum_{j \in \Jp} w_j \cdot T + |I \setminus \Ip| \leq m. $$
It clearly holds for $p = 0$. As in the proof of Claim~\ref{cl:Hall}, if it holds at any time during a phase, it will continue to hold during a later point of time in this phase. Now suppose the condition holds at some time $T$ during a phase $p$ and we go to phase $(p+1)$ at $T$. This happens because we reach Step~\ref{i} during this phase. 
We claim that $$\sum_{j \in \Jp} w_j \cdot T + |I \setminus \Ip| = \sum_{j \in J^{(p+1)}} w_j \cdot T + |I \setminus I^{(p+1)}|. $$
This easily follows from the fact that $\sum_{j \in J'} R_j = w(J') \cdot T = \sum_{j \in I'} L_j = |I'|, $ 
where $I'$ and $J'$ are as defined in Step~\ref{i}. Therefore the invariant continues to hold in phase $(p+1)$. 

Suppose we reach Step~\ref{ii} during phase $p$. Note that for every $j \in I \setminus \Ip$, we have $L_j = 1$. 
In this phase $\sum_{j \in \Jp} w_j \cdot T = \sum_{j \in \Ip} L_j. $ The condition in Step~\ref{ii} shows that 
this quantity is equal to $m - |I \setminus \Ip|$. Therefore, $\sum_{j \in I} L_j = m$. 
\end{proof}

Thus, we have shown that the quantities $z, L, R$ satisfy all the constraints in~\eqref{cp}. Now we prove their optimality. 

%-------------------------------------------------------
\subsection{Optimality}
%Let $\ell$ denote the index of the last phase -- note that it could end with case~(i) (e.g., when the graph $H^{(p+1)}$ becomes empty) or with case (ii). W

To prove optimality, we will define non-negative dual variables $\theta_j, \eta, \nu$ which satisfy the KKT conditions~\eqref{k1}--\eqref{k6}. We give some notations first. Let $\ell$ denote the index of the final phase (the algorithm could end because of Steps~\ref{i} or~\ref{ii}). For any phase $p$, let $J^\Dp$ denote $\Jp \setminus J^{(p+1)}$ (this set is same as $J'$ used in Step~\ref{i}). Define $I^\Dp$ similarly. Since $I^\Dp = \Gamma(J^\Dp)$ in the graph $\Hp$, there cannot be an edge in $H$ between $J^\Dp$ and $I^{(\Delta_{p'})}$ for some $p' > p$ (though there could be an edge between $J^{(\Delta_{p'})}$ and $\Ip$). In case $p = \ell$, define $J^\Dp$ and $I^\Dp$ as $\Jp$ and $\Ip$, respectively. 
Let $T_p$ denote the time at which phase $p$ {\em ends}. 

Now we define the dual variables:
\begin{itemize}
\item $\theta_j$: Let $j \in I^\Dp$, where either $p \neq \ell$, or $p = \ell$ but the last phase ends in Step~\ref{i}. 
Define $\theta_j$ to be $1/T_p$. If $j \in I^{(\ell)}$ and the phase $\ell$ ends in Step~\ref{ii}, then define 
$\theta_j$ to be 0. 
\item $\eta$: If the last phase $\ell$ ends in Step~\ref{ii}, define $\eta$ to be $1/T_{\ell}$.  Otherwise, define $\eta$ to be $0$.
\item $\nu_e$: If the end-points of $e$ belong to $J^\Dp$ and $I^\Dp$ for some phase $p$, then $\nu_e$ is defined to be 0. The only other possibility is that the end-points of $e$ belong to $J^{(\Delta_{p'})}$ and $\Ip$, respectively, where $p' > p$. In this case, define $\nu_e$ to be $1/T_p - 1/T_{p'}$. Clearly, $\nu_e \geq 0$ for all edges $e$. 
\end{itemize}

Checking KKT conditions is easy. To check~\eqref{k4}, note that if $\theta_j > 0$ then $j$ is assigned $L_j$ value in Step~\ref{i} of a phase, and so, $L_j = 1$.  
To check~\eqref{k5}, note that if $\eta > 0$ then we are in Step~\ref{ii} of the last phase, and so, Claim~\ref{cl:m2} shows that $\sum_{j \in I} L_j = m$. 
To check~\eqref{k6}, clearly if $\nu_e > 0$, then $z_e = 0$.
Finally, to check~\eqref{k1}, consider an edge $e = (j,j')$ with $j \in I^\Dp$ and $j' \in J^{(\Delta_{p'})}$  for some $p' \geq p$. Note that $\theta_j + \eta = \frac{1}{T_p}$ and $\frac{w_j}{R_j} = \frac{1}{T_{p'}}$. But then $\nu_e$ is exactly the difference between these two terms. 

Since the KKT conditions~\eqref{k1}--\eqref{k6} are satisfied, this proves the optimality of our algorithm.

%%% Local Variables: 
%%% mode: latex
%%% TeX-master: "main"
%%% End: 

%\section*{Acknowledgment}

\section{The Missing Proofs}
\label{sec:proofs}

\subsection{Proofs for Section~\ref{sec:compleTime}}

\begin{proof}[Proof of Claim~\ref{cl:lprelx}]
Consider an optimal schedule $\calS$, and let $x_{j,t}$ be the volume of
  $j$ processed during $[t,t+1]$. Constraint~(\ref{eq:1}) states that
  the total amount of processing on $j$ must be at least (in fact, it
  will be equal to) $p_j$. Constraint~(\ref{eq:2}) requires that the
  total amount of processing that can happen during a slot $[t,t+1]$ is
  at most $m$ because each machine can perform 1 unit of processing
  during this time slot. Constraint~(\ref{eq:3}) can be justified as
  follows: suppose $j$ precedes $j'$, and consider a time $t$. Then the
  LHS of this constraint denotes the fraction to which $j$ has been
  processed till time $t$, and the RHS denotes this quantity for
  $j'$. In the schedule $\calS$, if the RHS is positive, then it must be
  the case that $j$ has been completed by time $t$, and so the LHS would
  be 1.
  Finally, we consider the objective function. Let $C_j$ be the
  completion time of $j$. Clearly, $x_{j,t} = 0$ for $t > C_j$, and so,
  $\sum_t \frac{t \cdot x_{j,t}}{p_j} \leq C_j \cdot \frac{\sum_t
    x_{j,t}}{p_j} = C_j. $
\end{proof}

 \begin{proof}[Proof of Claim~\ref{cl:full}]
 Suppose $\sum_{j \in I_t} \btL_j < m$, but $\btL_j < 1$ for some
  $j \in I_t$. Let $e$ be an edge incident with $j$ (since there is a
  copy of $j$ on the right side of the bipartite graph, we know that $j$
  has at least one edge incident with it). We can raise the $z_e$ value
  of this edge while maintaining feasibility. But this will increase the
  objective value, a contradiction.
\end{proof}

\begin{proof}[Proof of Claim~\ref{cl:sum0}]
  Constraint~\eqref{k1} implies that $\btR_j >
  0$ and so there is a vertex $j' \in I_t$ such that $e=(j',j) \in
  H_t$ with $\btz_e > 0$. Now~\eqref{k6} shows that $\nut_e =
  0$, and so $w_j/\bR^t_j = \etat + \theta^t_{j'} \geq
  \etat$. Hence the proof.
\end{proof}

\begin{proof}[Proof of Claim~\ref{cl:sum}]
  Let $e$ denote the edge $(j',j)$. Now~\eqref{k4} and~\eqref{k6} imply
  that $\nut_e = 0$ and $\thetat_{j'} = 0$. The claim now follows
  from~\eqref{k1}.
\end{proof}

\begin{proof}[Proof of Corollary~\ref{cor:sum}]
 By definition there is a neighbor $j' \in I_t$ of $j$ such that
  $\btL_{j'} < 1$.  Let $e'$ denote the edge $(j',j)$. If
  $\btz_{e'} > 0$, then we are done by Claim~\ref{cl:sum} above. So assume
  $\btz_{e'} = 0$.  Since $\btR_j > 0$. There must be an edge
  $e''=(j'',j)$ incident with $j$ such that $\btz_{e''} > 0$.  Again, if
  $\btL_{j''} < 1$, we are done by the Claim above. So, assume that
  $\btL_{j''} = 1$. Now consider reducing $\btz_{e''}$ by a tiny amount
  and increasing $\btz_{e'}$ by the same amount. This maintains
  feasibility of all constraints. Since $\btR_j$ remains unchanged, we
  remain at an optimal solution. Now we can apply Claim~\ref{cl:sum}.
\end{proof}

\begin{proof}[Proof of Claim~\ref{cl:upper}]
  Let us prove the upper bound first. If $\etat = 0$, there is nothing
  to prove. So assume $\etat > 0$. Constraint~\eqref{k5} now implies
  that
  \begin{gather}
    \ts \etat = \etat \cdot \nicefrac1m \cdot \sum_{j \in I_t} \btL_j =
    \etat \cdot \nicefrac1m \cdot
    \sum_{j \in J_t} \btR_j, \label{eq:4}
  \end{gather}
  the latter using~\eqref{eq:c1} and~\eqref{eq:c2}. Now using
  Claim~\ref{cl:sum0}, we can bound $\etat \btR_j \leq w_j$
  in~(\ref{eq:4}), giving us $\etat \leq \frac1m \sum_{j \in J_t} w_j$, and hence
  the upper bound.
  % Combining this
  % with~\eqref{k1}, we get
  % $$\etat = \etat \cdot \frac{\sum_{j \in J_t} w_j/b^t_j}{m} . $$
  % For the second statement,
  % Consider a job $j \in J_t$. Since $\btR_j > 0$ by~\eqref{k1}, there
  % is an edge $e=(j',j)$ incident with $j$ such that $\btz_e > 0$. This
  % implies that $\nut_e = 0$ (using~\eqref{k6}). From~\eqref{k2}
  % and~\eqref{k3}, we know that
  % $b^t_{j} = \etat + \theta^t_{j'} - \nut_e = \etat + \theta^t_{j'}
  % \geq \etat$. Thus, we get
  % $$\etat \leq \frac{\sum_{j \in J_t}w_j}{m}. $$
  For the lower bound, suppose $\eta_t = 0$. Then for every job $j \in
  J_t$ and for every edge $(j',j) \in E_t$, we must have $\thetat_{j'} >
  0$. This means each of the jobs in $J_t$ are inactive, and hence
  $w(\Jact_t) = 0$, which proves the claim. The other case is when
  $\eta_t > 0$, and then we get: 
  \begin{gather*}
   \ts \etat \stackrel{(\ref{eq:4})}{=} \etat \cdot \nicefrac1m \cdot
    \sum_{j \in J_t} \btR_j \geq \etat \cdot \nicefrac1m \cdot
    \sum_{j \in \Jact_t} \btR_j = w(\Jact_t)/m,
  \end{gather*}
  where the last equality follows from Corollary~\ref{cor:sum}. 
\end{proof}

 \begin{proof}[Proof of Claim~\ref{cl:obj}]
 Fix a job $j$. Let $C$ be the chain in $G$ which ends with $j$ and satisfies $p(C) = \chain_j$. Consider a time $t \leq C_j$, the completion time of $j$. Suppose $\alpha_{j,t} = 0$. Considering $j$ as a vertex in $J_t$ (i.e., right side) in the bipartite graph $H_t$, it must be the case that all its neighbors get rate 1. Exactly one job in the chain $C$, say $j'$, belongs to the set $I_t$. Since $(j',j)$ is an edge in $H_t$, it must be the case that $j'$ gets rate 1. Thus, we conclude that whenever $\alpha_{j,t} = 0$, there is a job in $C$ which is processed for 2 units during $[t,t+1]$ (recall that the machines in $\A$ run at speed 2). Therefore, $w_j(C_j - r_j) \leq \alpha_j + w_j \cdot \chain_j/2$. Summing over all jobs, we get
$$ \sum_j w_j C_j \leq \sum_j w_j (r_j + \chain_j/2) + \sum_j \alpha_j. $$
Now observe that for any time $t$, $m \beta_t$ is equal to $w(U_t)/2$, and so, $m \cdot \sum_t \beta_t = \sum_j w_j C_j/2$. Subtracting this from the inequality above yields the desired result.
\end{proof}

 \subsection{Proofs for Section~\ref{sec:flow}}
 
 \begin{proof}[Proof of Claim~\ref{cl:objf}]
 Suppose $t$ is a time at which all machines are busy (i.e.,
  $\sum_{j \in I_t} \btL_j = m$). We first argue that
  $\sum_{j \in J_t} \alpha_{j,t} $ is equal to $w(\Jact_t)$. Indeed, observe that $w_j \btR_{j'}$ appears in either $\alpha_{j,t}$ or
  $\alpha_{j',t}$ depending on whether $j' \preceq j$ or
  otherwise. Hence, we get
  \begin{align}
    \label{eq:objf}
    % \sum_j \alpha_j \quad = \quad
    \sum_{j \in J_t} \alpha_{j,t}  = \frac{1}{m} \sum_{j \in J_t} w_j
    \cdot\ind{j \in \Jact_t} \cdot    \sum_{j' \in J_t} \btR_j =
    \frac{1}{m} \sum_{j \in \Jact_t} w_j  \cdot    \sum_{j' \in J_t}
    \btL_j  =   w(\Jact_t).  
  \end{align}

  We now argue that
  $$\sum_t w(J_t \setminus \Jact_t) = \sum_j w_j \cdot \chain_j/(2+2\eps).$$ 
  Indeed, consider a job $j \in J_t \setminus \Jact_t$. All its
  neighbors in $G_t$ are running at rate 1. Therefore, we must be
  running a job in the chain which defines $\chain_j$. The factor
  $2(1+\eps)$ comes from the machine speedup. Observe that if all
  machines are not completely busy at time $t$, then all jobs in $J_t$
  are inactive~(Claim~\ref{cl:full}). Combining this
  with~\eqref{eq:objf}, we see that the total weighted flow-time is
  $$ \sum_t w(\Jact_t)  + \sum_t w(J_t \setminus \Jact_t) =   \sum_j
  \alpha_j + \sum_j w_j \cdot \chain_j/(2+2\eps). $$ The claim follows
  because $m \cdot \sum_t \beta_t$ is $1/(1+\eps)$ times the total
  weighted flow-time, which means the difference
  \[ \sum_{j} \alpha_j - m \sum_t \beta_t + \sum_j w_j \cdot \chain_j 
  \quad \geq \quad \Big(1 - \frac{1}{1+\eps}\Big) \cdot \sum_t w(J_t) \quad \geq \quad \frac{\eps}{2}  \cdot \sum_t w(J_t).\] 
 This finishes the proof of the claim because $\sum_t  w(J_t)$ is the total
  weighted flow-time.
\end{proof}

\begin{proof}[Proof of Claim~\ref{cl:sf1}]
  We begin by bounding $2(\gout_{s,j} - \gin_{s,j})$. As in the proof of
  Claim~\ref{cl:sf}, the contribution from paths $P_e$ for which $j$
  lies on the right side of $H_s$ is $-2 \eta^s_j \bR^s_j$, which by
  Lemma~\ref{lem:afnew} cancels $\alpha_{j,s}$. Thus we get
  $$ \alpha_{j,s} + 2(\gout_{s,j} - \gin_{s,j}) = 2 \sum_{j' \in J_s :
    e=(j \rightarrow j') \in E_s} \eta^s_{j'} \cdot \bz^s_{e}. $$
  Finally, recall that all jobs in a DAG have the same release
  time. Hence, any job $j'$ in the summation above is released at the
  same time as $j$. Moreover, any job $j'' \in J_s$ which contributes
  towards
  $\eta^s_{j'} = \frac1m \cdot ( \sum_{j'' \in J_s: j'' \preceq j'}
  w_{j''} )$ has  been  also released at or before $r_j$. Therefore,
  $\eta^s_{j'} \leq w(J_t)/m = (1+\eps)\beta_t$ by definition of $\beta_t$. 
  This implies 
  \[ \alpha_{j,s} + 2(\gout_{s,j} - \gin_{s,j}) \quad \leq \quad 2 \sum_{j' \in J_s :
    e=(j \rightarrow j') \in E_s} (1+\eps)\beta_t \cdot \bz^s_{e} \quad = \quad
     2(1+\eps) \beta_t
  \cdot \bL^s_j,
  \]
  where we use $\bL^s_j = \sum_{j' \in J_s :
    e=(j \rightarrow j') \in E_s} \bz^s_{e}$.
\end{proof}

%%% Local Variables:
%%% mode: latex
%%% TeX-master: "main"
%%% End:

{\small
\bibliographystyle{alpha}
\bibliography{sched}
}

\end{document}